\documentclass[journal]{IEEEtran}
\usepackage{amsmath,amsfonts}
\usepackage{algorithmic}
\usepackage{algorithm}
\usepackage{array}
\usepackage[caption=false,font=normalsize,labelfont=sf,textfont=sf]{subfig}
\usepackage{textcomp}
\usepackage{stfloats}
\usepackage{url}
\usepackage{verbatim}
\usepackage{graphicx}
\usepackage{cite}

\usepackage{amsthm}
\newtheorem{theorem}{Theorem}
\newtheorem{assumption}{Assumption}
\newtheorem{lemma}{Lemma}
\usepackage{tabularx}
\usepackage{booktabs}
\usepackage{adjustbox}
\newcolumntype{Y}{>{\centering\arraybackslash}X}
\usepackage{colortbl}
\usepackage{xcolor}
\usepackage{amssymb}

\begin{document}

% \title{CacheFL: Efficient Federated Cache Model Fine-Tuning for Vision-Language Models}
% \title{Privacy-Preserving Yet Efficient: Federated Cache Model Fine-Tuning for Vision-Language Models}
\title{CacheFL: Privacy-Preserving and Efficient Federated Cache Model Fine-Tuning for Vision-Language Models}

\author{Mengjun Yi,
        Hanwen Zhang,
        Hui Dou,
        Furao Shen,~\IEEEmembership{Member,~IEEE} and Jian Zhao,~\IEEEmembership{Senior Member,~IEEE}
        % <-this % stops a space
\thanks{Mengjun Yi, Hanwen Zhang, and Furao Shen are with National Key Laboratory for Novel Software Technology and School of Artificial Intelligence, Nanjing University, Nanjing 210023, China (e-mail: mengjunyi@smail.nju.edu.cn; hwzhang@smail.nju.edu.cn; frshen@nju.edu.cn). \textit{(Corresponding author: Furao Shen.)}}
\thanks{Hui Dou is with School of Computer Science, Nanjing University, Nanjing 210023, China (e-mail: huidou@smail.nju.edu.cn). }
\thanks{Jian Zhao is with School of Electronic Science and Engineering, Nanjing University, Nanjing 210023, China (e-mail: jianzhao@nju.edu.cn). }}% <-this % stops a space

\markboth{Journal of \LaTeX\ Class Files,~Vol.~14, No.~8, August~2021}%
{Shell \MakeLowercase{\textit{et al.}}: A Sample Article Using IEEEtran.cls for IEEE Journals}


\maketitle

\begin{abstract}
Large pre-trained Vision-Language Models (VLMs), such as Contrastive Language-Image Pre-training (CLIP), have exhibited remarkable zero-shot performance across various image classification tasks. Fine-tuning these models on domain-specific datasets further enhances their effectiveness for downstream applications. However, fine-tuning in cloud environments raises significant concerns regarding data security and privacy. Federated Learning (FL) offers a decentralized solution by enabling model training across local clients without centralizing sensitive data, but the high communication and computation costs of transmitting full pre-trained models during training limit its scalability. Additionally, non-Independent and Identically Distributed (non-IID) data across local clients can negatively impact model convergence and performance. To address these challenges, we propose CacheFL, a novel federated learning method that replaces traditional full model fine-tuning with lightweight cache model fine-tuning. The cache model is initialized using a class-balanced dataset generated by a generative pre-trained model, effectively mitigating the impact of non-IID data. This cache model is then distributed to local clients for fine-tuning, and the updated parameters from each client are aggregated on the server and redistributed. With the updated cache model, the classification performance of CLIP is improved after just a few epochs. By limiting the training and communication to the cache model, CacheFL significantly reduces resource demands while ensuring data privacy and security. Extensive experiments conducted on ImageNet and 10 additional datasets demonstrate that CacheFL outperforms traditional approaches in terms of classification accuracy, resource efficiency, and privacy preservation.
\end{abstract}

\begin{IEEEkeywords}
Edge computing, federated learning, vision-language model.
\end{IEEEkeywords}

\section{Introduction}

\begin{figure}[t]
  \centering
  \includegraphics[width=0.85\linewidth]{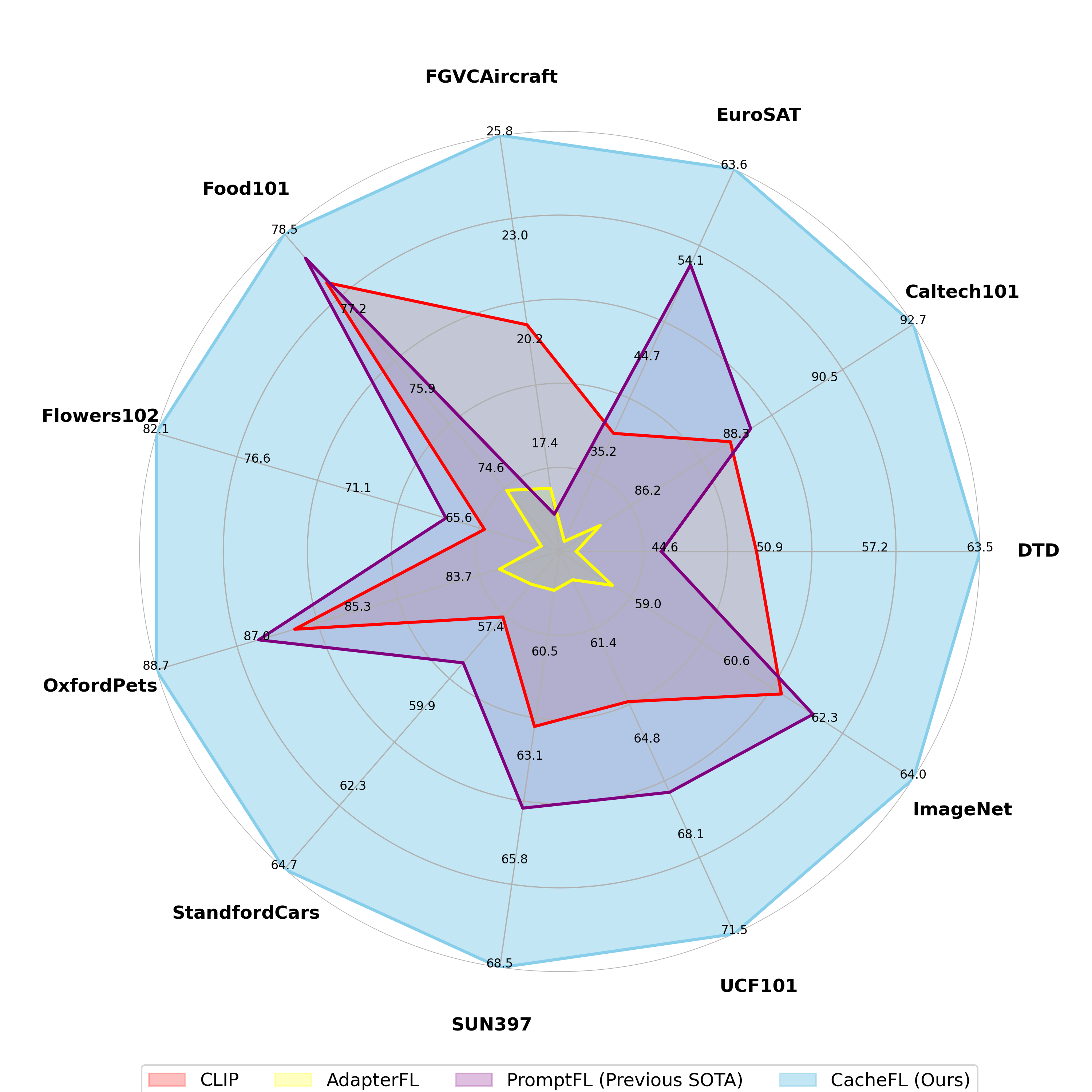}
  \caption{Performance comparison on the Extreme non-IID setting. Each axis denotes an accuracy value for the corresponding dataset (refer to Table \ref{tab:results}). CacheFL surpasses state-of-the-art methods on 11 diverse image classification datasets.}
  \label{fig:pc}
\end{figure}

\IEEEPARstart{W}{ith} the rapid advancement of artificial intelligence, Vision-Language Models (VLMs) \cite{zhang2024vision} have emerged as a key technology in multimodal learning, garnering significant attention in recent years. VLMs simultaneously model visual and linguistic signals, enabling cross-modal understanding between images and text. Among these, Contrastive Language–Image Pretraining (CLIP) \cite{radford2021learning} has emerged as a representative and influential approach. By leveraging large-scale image-text pair datasets and employing contrastive learning techniques, CLIP demonstrates remarkable performance in tasks such as image classification. While CLIP exhibits strong generalization capabilities, its performance can often be further enhanced through fine-tuning on domain-specific datasets \cite{khattak2023maple}.

However, effective fine-tuning often requires access to diverse datasets distributed across various organizations or clients. Due to privacy, security, or storage constraints, these datasets are typically challenging to centralize in the cloud for training \cite{zhou2019edge}. To overcome this limitation, Federated Learning (FL) \cite{kairouz2021advances, liu2021adaptive, wu2020fedhome} has emerged as a promising paradigm for privacy-preserving and collaborative model optimization. FL enables multiple clients to train models locally on their private data and share updates to build a global model, ensuring that data remains ``usable but invisible" during the distributed training process.

Despite its potential, directly applying FL to fine-tune CLIP and similar vision-language models presents two significant challenges. First, the large parameter sizes of these models create substantial communication and computation overheads in the federated learning framework \cite{chen2021communication, zhang2024heterogeneity, jiang2023computation}. Second, the presence of non-Independent and Identically Distributed (non-IID) client data often results in slow convergence and limited performance gains in the global model \cite{li2022federated, wang2024fedsiam, xu2024overcoming}. 

To address these challenges, we propose \textbf{CacheFL}, a cache model-based federated learning fine-tuning method designed to optimize the fine-tuning process of CLIP and similar vision-language models in federated settings. CacheFL begins by generating a class-balanced synthetic dataset on the server using a generative pre-trained model \cite{ramesh2021zero}. Inspired by the TIP-Adapter method \cite{zhang2021tip}, we construct a lightweight cache model by storing visual features and
one-hot encoded labels derived from the synthetic dataset. This cache model is then distributed to the clients. In contrast to fine-tuning the entire model, transmitting and training only the cache model significantly reduces both communication and computation overheads. Each client updates the cache model using its local data, and the parameter updates are subsequently aggregated on the server and redistributed to the clients. Through this iterative process, CacheFL enhances CLIP's performance (see Figure \ref{fig:pc}) while preserving training efficiency. The main works of this paper include the following:
\begin{itemize}
\item \textbf{Efficient.} We propose CacheFL, a novel federated cache model fine-tuning framework for vision-language models. The cache model comprises only visual features and one-hot encoded labels, forming a lightweight structure with significantly fewer parameters. By replacing conventional full model fine-tuning with cache model fine-tuning, CacheFL greatly reduces both communication and computation overhead.

\item \textbf{Privacy-Preserving.} We adopt federated learning in place of centralized fine-tuning and introduce a new cache model initialization strategy. 
This design allows vision-language models to be fine-tuned without exposing any local data. Furthermore, we utilize a generative pre-trained model to synthesize a class-balanced dataset, enabling cache model initialization without relying on additional public or private datasets. This approach not only reduces data collection costs but also safeguards data privacy.

\item \textbf{Effective.} We conduct experiments on various image classification datasets under both IID and non-IID settings. Experimental results show that our method excels in model accuracy, resource efficiency, and privacy preservation. Moreover, CacheFL effectively mitigates the challenges posed by non-IID data in federated learning scenarios.
\end{itemize}

The remainder of this paper is organized as follows. Section~\ref{2} reviews related work on vision-language models and federated learning. Section~\ref{3} describes the proposed method in detail. Section~\ref{6} provides the convergence analysis, and Section~\ref{4} presents experimental results that demonstrate the effectiveness of the approach.  Finally, Section~\ref{5} concludes the paper and discusses potential future research directions.

\section{Related Work}\label{2}
In this section, we review the foundational concepts, key methodologies, and relevant literature on vision-language models and federated learning. 

\subsection{Vision-Language Models}
Vision-language models have demonstrated remarkable potential in the field of multimodal learning in recent years \cite{zhou2022learning}. These models leverage joint learning of visual and linguistic features to tackle tasks such as image classification, image-text matching, and text generation. 

\begin{figure}[t]
  \centering
  \includegraphics[width=\linewidth]{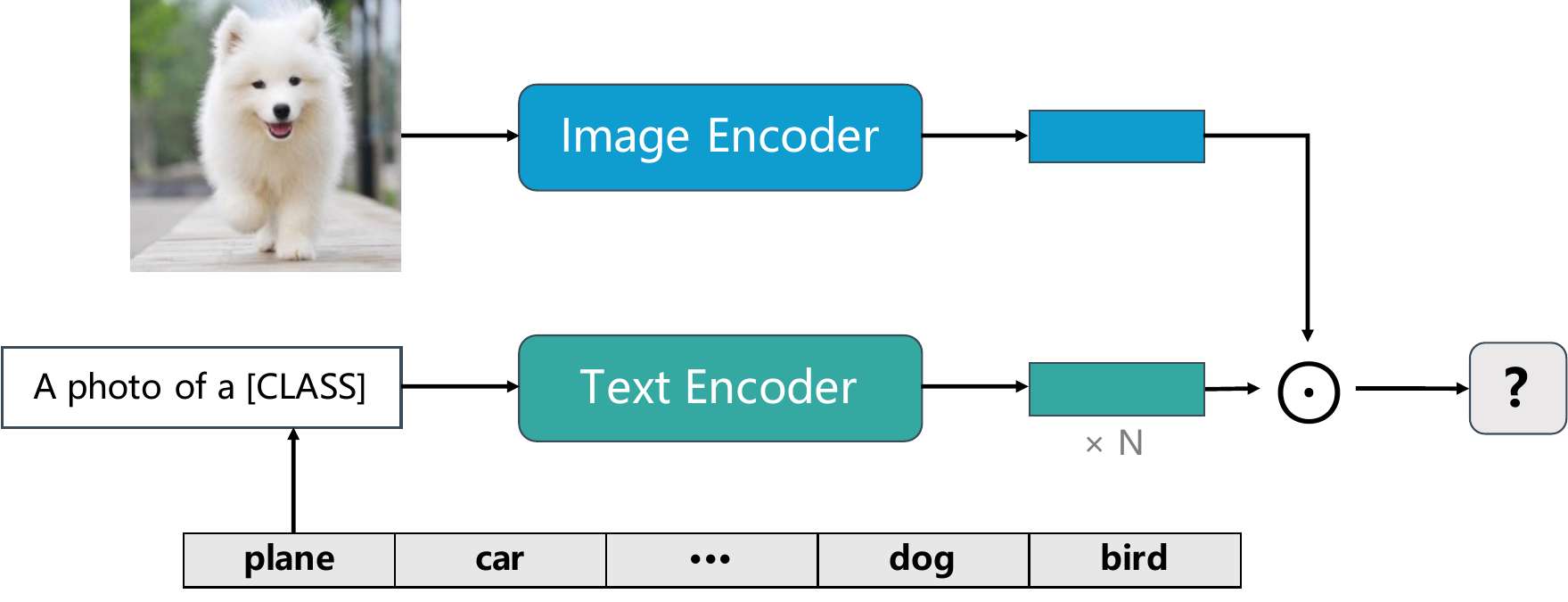}
  \caption{The Structure of CLIP. Each backbone contains an image encoder and a text encoder. Each encoder extracts feature representations from its respective modality. The model is trained to maximize the cosine similarity of the correct image-text pairs while minimizing the incorrect pairs.}
  \label{fig:clip}
\end{figure}

A representative work in this domain is CLIP proposed by OpenAI \cite{radford2021learning}. CLIP is a significant advance in multimodal machine learning that connects natural language and images. 
As shown in Figure \ref{fig:clip}, CLIP comprises two main components: Image Encoder and Text Encoder. The image encoder is often a convolutional neural network (CNN) or a vision transformer (ViT) \cite{dosovitskiy2020image}. It takes an image input and maps it to a fixed-dimensional feature vector, or an embedding, in a latent space. The text encoder is typically a transformer-based architecture \cite{vaswani2017attention}. It takes a natural language input (a sentence or phrase) and similarly maps it to an embedding in the same latent space. 
The goal is for the corresponding image and text (e.g., a photo of a dog and the caption “a dog”) to have similar embeddings, while unrelated images and texts should be far apart in this space. CLIP is notable for its zero-shot prediction capabilities, meaning it can generalize to new tasks without explicit task-specific training. For classification, CLIP encodes an image into an embedding. Simultaneously, a set of possible class labels (e.g., ``dog", ``cat", ``car") is encoded by the text encoder into corresponding embeddings. The model then compares the similarity between the image embedding and each class embedding. The class with the highest similarity score is selected as the prediction. 

While VLMs like CLIP excel in zero-shot learning scenarios, their performance can often be further enhanced through fine-tuning with task-specific data. Common fine-tuning methods can broadly be categorized into full fine-tuning and parameter-efficient fine-tuning \cite{ding2023parameter}. Full fine-tuning updates all model parameters to achieve optimal performance, but often requires substantial computational resources and may lead to overfitting, especially with limited data. In contrast, parameter-efficient fine-tuning adjusts only a subset of the model's parameters, thereby reducing resource consumption. Adapter tuning \cite{he2021effectiveness} is a parameter-efficient fine-tuning method involving inserting additional small neural network modules, known as adapters, into the structure of a pre-trained model. During fine-tuning, only the parameters of these adapters are trained, while most of the pre-trained model's parameters remain frozen. The CacheFL method proposed in this paper also belongs to an adapter tuning method. Additionally, prompt tuning \cite{lester2021power} is another method of parameter-efficient fine-tuning that guides the outputs of the pre-trained model by embedding task-specific prompt information into the input data. These prompts are either manually designed or trained using domain-specific datasets. However, regardless of the fine-tuning method employed, they often rely on downstream task data, which raises concerns about data security and privacy.

\subsection{Federated Learning}
Federated learning is a distributed machine learning framework that allows participants to collaboratively train models without sharing raw data, addressing privacy and security concerns in centralized learning \cite{li2020review}. Data from different organizations, such as hospitals, cannot be centralized due to legal and regulatory constraints. Federated learning enables these entities to jointly train models while maintaining data privacy.

In federated learning, participants can process data and train models locally, transmitting only encrypted model parameters to a server for aggregation. The specific steps are as follows: first, the server initializes a global model and distributes it to participant devices. Subsequently, each device trains the model on its local data and updates its local model weights. Each device then sends its local model updates back to the server. The server receives updates from multiple devices and employs a specific aggregation algorithm to consolidate these updates into a new global model. This process is repeated until the model converges or meets predetermined performance criteria. Typically, the server uses the Federated Averaging algorithm (FedAvg) \cite{mcmahan2017communication} to perform weighted average updates of the model parameters from each client. 
% After several local training iterations, each client uploads its model parameters to the server, which then computes a weighted average based on the data volume from each client to update the global model. 
The FedAvg formula is as follows:
\begin{equation}
W^{t+1}=\sum_{k=1}^{K} \frac{n_k}{n}  W_{k}^{t+1},
\label{eq:fedavg}
\end{equation}
where $W^{t+1}$ represents the global model parameters at iteration $t+1$, $K$ is the number of participating clients, $n_k$ represents the number of samples on the $k$-th client, $n$ represents the total number of samples on all clients and $W_{k}^{t+1}$ denotes the local model parameters of the $k$-th client at iteration $t+1$. 

However, FL faces significant challenges when applied to resource-constrained clients, especially with large-scale models such as VLMs. These models often contain millions or even billions of parameters, leading to excessive communication and computation overhead during local training and model update transmission. To address these challenges, researchers have proposed integrating parameter-efficient fine-tuning techniques with FL. For instance, PromptFL \cite{guo2023promptfl} is a method that combines prompt tuning with federated learning, allowing multiple devices to fine-tune prompts instead of the CLIP model locally. Another significant challenge in FL is the non-IID data among the clients~\cite{li2022federated}. Heterogeneous data has a significant impact on the accuracy of FedAvg. Since the distribution of each local dataset may be different from the global distribution, the local objective of each client may not be consistent with the global optimal solution. As a result, there are drifts in the local updates. When local updates are large, the average model may be far from the global optimal solution. Li \emph{et al.} \cite{li2020federated} proposed regularization techniques to mitigate data drift. However, these methods are challenging to apply directly to pre-trained models due to high resource demands. Thus, it is imperative to develop novel solutions that simultaneously address the dual challenges of resource constraints and non-IID data for the integration of FL with VLMs.

\begin{figure}[t]
  \centering
  \includegraphics[width=0.85\linewidth]{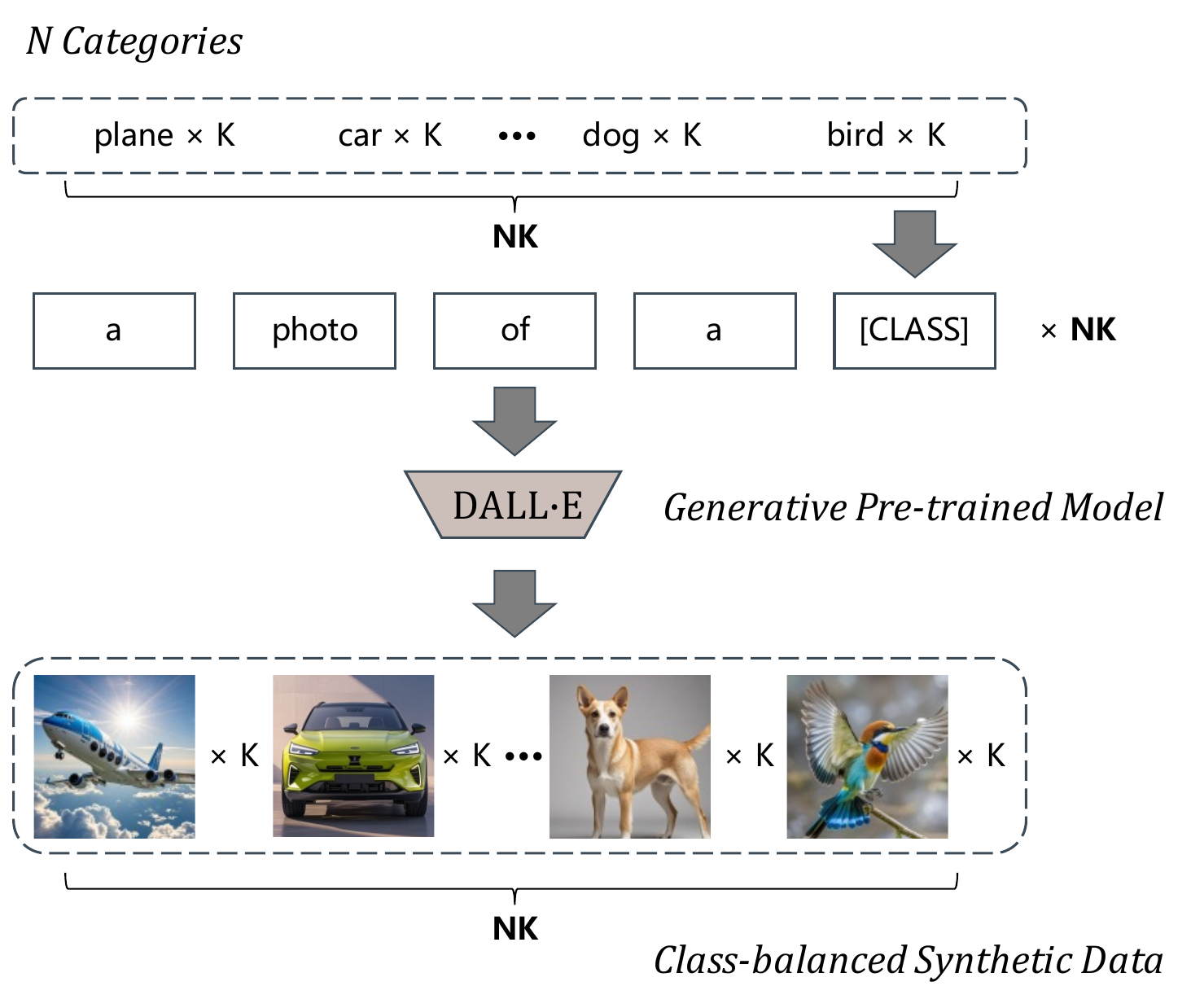}
  \caption{Class-balanced Synthetic Dataset Generation. We adopt DALL·E to generate $K$ synthetic images for $N$ categories separately.}
  \label{fig:syn}
\end{figure}

\section{Methodology}\label{3}

\begin{figure*}[t]
  \centering
  \includegraphics[width=0.9\linewidth]{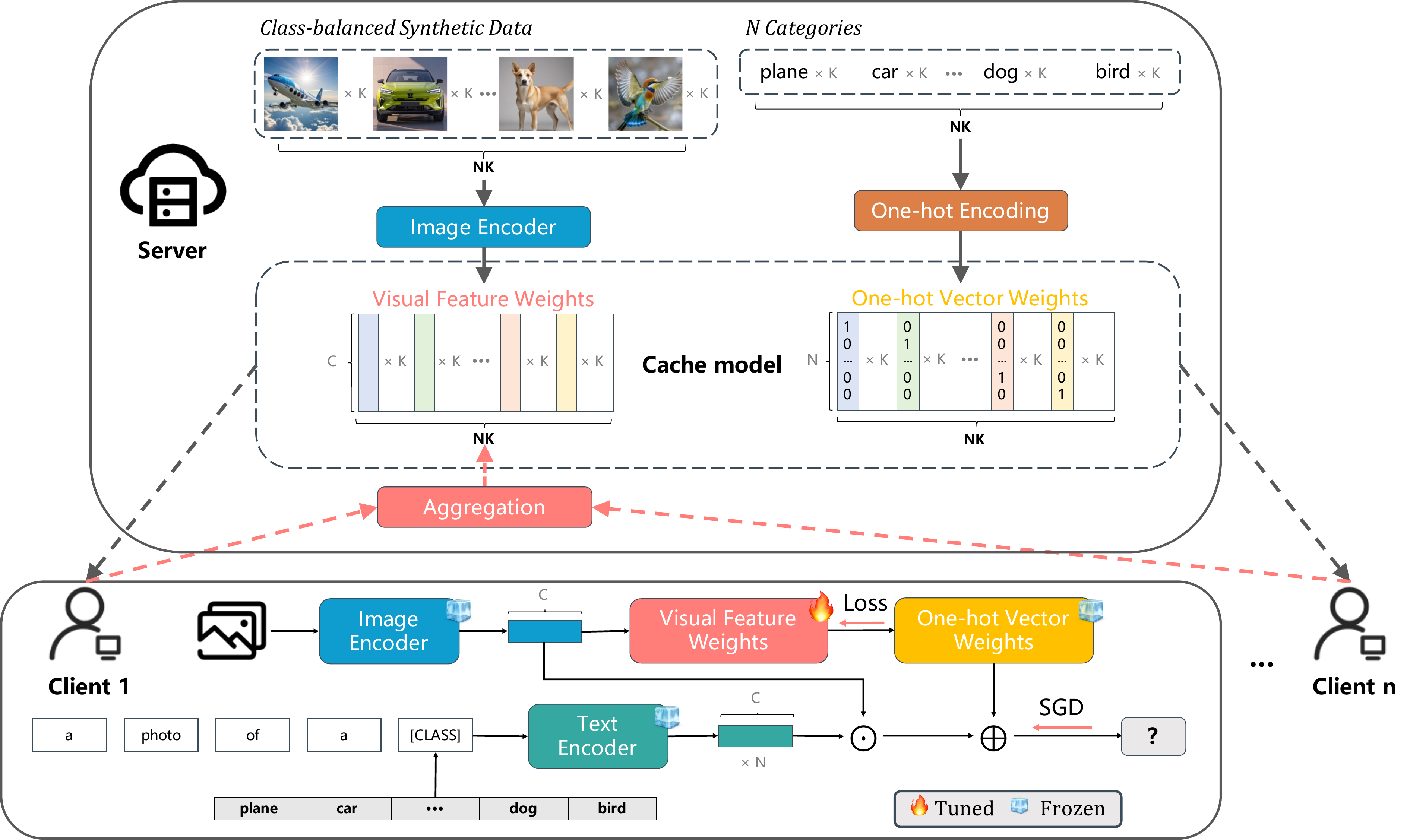}
  \caption{The workflow of CacheFL Federated Training. Each client includes a cache model (with the visual feature weights being trainable and the one-hot vector weights remaining frozen) and an out-of-the-box CLIP (with the backbone remaining frozen). The server aggregates the updates to the visual feature weights from multiple clients, and then transmits the updated parameters back to each client for further local processing.}
  \label{fig:cache}
\end{figure*}

% \subsection{Federated Cache Model Fine-Tuning}
In this section, we provide the details of our proposed method, CacheFL. To further enhance the classification performance of CLIP and its similar structures, CacheFL introduces a three-step process: (1) A class-balanced synthetic dataset is generated on the server using a generative pre-trained model. (2) A lightweight cache model is built based on this synthetic dataset. (3) The cache model is distributed to clients, where local fine-tuning is performed on their private datasets. After local training, the updated model parameters are returned to the server and aggregated to optimize the global model. Detailed descriptions of each step are provided below.

\subsection{Class-Balanced Synthetic Dataset Generation} 
Given the concerns regarding data security and privacy, the server often lacks sufficient high-quality data for model training. Moreover, collecting public datasets is costly and may not ensure balanced class distributions. While Generative Adversarial Networks (GANs) \cite{goodfellow2020generative} are widely used for data generation, their sample quality is often insufficient due to issues like mode collapse \cite{thanh2020catastrophic}, making them less practical for real-world use. In contrast, recent generative pre-trained models have shown greater potential in data generation, addressing challenges related to data scarcity and quality.

We employ the DALL·E generative pre-trained model from OpenAI for data generation. As illustrated in Figure \ref{fig:syn}, we input N class labels formatted as “a photo of a [CLASS]” into the DALL·E model, generating K images for each label, expressed as: 
\begin{equation}
    I_{N,K} = DALL \cdotp E (T_N),
\end{equation}
where $T_N$ represents the textual input of the $N$ class labels, and $I_{N,K}$ represents the $NK$ images generated for these labels. 
This method automatically generates a synthetic dataset with balanced class distributions, greatly reducing the need for manual data collection while ensuring both diversity and balance. By mitigating potential biases caused by data scarcity or imbalance, it enhances the model’s generalization capability in federated learning. 

Only a small amount of synthetic data per category needs to be generated using DALL·E before training. This process can be performed during server idle times to optimize resource usage. Alternatively, in resource-constrained environments, the synthetic dataset can be obtained via external APIs.

\subsection{Cache Model Construction} 
Directly fine-tuning all parameters of large pre-trained models like CLIP, though improving performance, demands enormous communication and computation resources. In contrast, adapter tuning provides a more lightweight fine-tuning approach. Most of the original model's parameters remain frozen, and only the adapter layers are trained, thus reducing overhead.

We use the class-balanced synthetic dataset to initialize a cache model, which serves as the adapter for CLIP. As shown in Figure \ref{fig:cache}, we extract visual features from the synthetic dataset of NK images using CLIP's image encoder, where the feature dimension is $C$, expressed as: 
\begin{equation}
    W_1 = ImageEncoder(I_{N,K}),
\end{equation}
where $W_1$ is a matrix with $C$ rows and $NK$ columns, representing cached visual feature weights. Next, we convert the labels $L_{N,K}$ of the $NK$ images into one-hot encoded vectors, expressed as: 
\begin{equation}
    W_2 = OneHot(L_{N,K}),
\end{equation}
where $W_2$ is a matrix with $N$ rows and $NK$ columns, representing cached one-hot vector weights. 
The visual feature weights $W_1$ and the one-hot vector weights $W_2$ together constitute a cache model that stores precomputed features and their associated labels, derived from synthetic images. By caching this information, the model enables efficient retrieval during both training and inference. This architecture is referred to as a ``\textbf{cache model}" because it encapsulates and reuses prior knowledge from a class-balanced dataset, serving as an effective initialization for subsequent federated fine-tuning.

\begin{algorithm}[ht]
	\caption{CacheFL}
    \label{fedavg1}
    \begin{algorithmic}[1]
        
    \STATE \textbf{SERVER:}
    \STATE Server initializes the cache model $W_1,W_2$ by the synthetic dataset
    \STATE Server broadcasts the cache model $W_1,W_2$ to all clients
    \FOR {each communication round t = 1, 2, ..., T}
        \STATE $\mathcal{S}_t \leftarrow$ Server selects a random subset of $K$ clients
        \STATE Server broadcasts $W_1^t$ to all selected clients
        \FOR {each client $k \in \mathcal{S}_t$ parallelly}
            \STATE $W_{1,k}^{t+1} \leftarrow ClientUpdate(k, W_1^t)$
        \ENDFOR
        \STATE $W_1^{t+1} \leftarrow \sum_{k \in \mathcal{S}_t} \frac{n_{k}}{n} W_{1,k}^{t+1}$
    \ENDFOR
    \STATE
    \STATE \textbf{CLIENT:} $ClientUpdate(k, W_1^t)$
    \FOR {each local epoch $e = 1, 2, ..., E$ on local data $D_k$}
        \STATE Obtain visual features: \\
        $f_{\text {vision}} = ImageEncoder(D_k)$ \\  
        \STATE Obtain textual features: \\
        $W_{\text {text}} = TextEncoder(Input_{N})$ \\ 
        \STATE Obtain the predicted logits: \\ 
        $logits = f_{\text {vision}} W_{\text {text}}^{T} + \alpha \exp \left(-\beta\left(1-f_{\text {vision}} W_1\right)\right) W_2^{T}$ \\  
        \STATE Update visual feature weights by cross-entropy loss $\mathcal{L}(\cdot)$ and SGD: \\
        $W_{1,k}^{t+1} \leftarrow W_1^t - \eta \nabla \mathcal{L}(W_1^t;(logits,y))$ \\
    \ENDFOR
    \end{algorithmic}
\end{algorithm}

\subsection{Federated Training} 
While the cache model based on the synthetic dataset can enhance downstream task performance to some degree, the improvement is often limited. Leveraging real-world data distributed across clients can significantly boost the model's performance. 

To protect data privacy while using local data, we apply federated learning to optimize the cache model's update process. As described in Algorithm \ref{fedavg1}, after initializing the cache model on the synthetic dataset, the server distributes it to each client. On the client side, images from local datasets are input into CLIP's image encoder to generate $C$-dimensional image feature vectors $f_{\text {vision}}$. Simultaneously, the descriptive text ``a photo of a [CLASS]" for $N$ classes is input into CLIP's text encoder to generate $N$  $C$-dimensional text feature vectors $f_{\text {text}}$, forming an $N \times C$ matrix $W_{\text {text}}$. Next, the logits between the image feature vector and the class text feature vectors are computed as follows: 
\begin{equation}
logits_1=f_{\text {vision}} W_{\text {text }}^{T}.
\end{equation}

Additionally, following the Tip-Adapter method \cite{zhang2021tip}, the logits for each image are computed using the cache model as:
\begin{equation}
logits_{2}=\exp \left(-\beta\left(1-f_{\text {vision}} W_{1}\right)\right) W_{2}^{T},
\end{equation}
where $\beta$ is a modulating hyperparameter. This step is equivalent to calculating the Euclidean distance between the feature $f_{\text {vision}}$ and the visual feature weights $W_1$ from the synthetic dataset and then querying the corresponding value.

Finally, $logits_{1}$ from the pre-trained CLIP and $logits_{2}$ according to values retrieved from the cache model are fused to obtain the final prediction logits:
\begin{equation}
    logits = logits_{1} + \alpha \cdot logits_{2},
\end{equation}
where $\alpha$ is a hyperparameter that can be adjusted depending on the task requirements. When there is a large discrepancy between the downstream task and the synthetic dataset, a smaller $\alpha$ is preferred, whereas a larger $\alpha$ is suitable when the differences are smaller.

The cache model is optimized using cross-entropy loss $\mathcal{L}(\cdot)$ and the Stochastic Gradient Descent (SGD) algorithm. During $t$-th round of local optimization, the one-hot vector weights $W_2$ remain frozen, and only the visual feature weights $W_1$ are updated on client $k$ as follows:
\begin{equation}
    W_{1,k}^{t+1} \leftarrow W_1^t - \eta \nabla \mathcal{L}(W_1^t;(logits,y)).
\end{equation}
After completing the $t$-th round of local training, client $k$ sends the updated parameters $W_{1,k}^{t+1}$ back to the server for aggregation, similar to Eqn.(\ref{eq:fedavg}) of FedAvg, as follows:
\begin{equation}
    W_1^{t+1} \leftarrow \sum_{k \in \mathcal{S}_t} \frac{n_{k}}{n} W_{1,k}^{t+1}.
\end{equation}

% Since the cache model has been well initialized on the class-balanced synthetic dataset, it quickly converges with just a few epochs on the client side, significantly improving the model's performance in real-world scenarios.

\section{Theoretical Analysis}\label{6}
In this section, we do a thorough analysis of the convergence of the CacheFL algorithm. To enhance readability and avoid the excessive complexity of notation, we use a simplified set of symbols for use throughout this section, which may differ from those employed in other parts of the paper.

\subsection{Definition of CacheFL}
In this work, we consider the objective function of CacheFL:
\begin{align}
    \min_{c}\left\{F(c)\doteq \sum_{k=1}^N p_k F_k(c) \right\},
\end{align}
where $N$ represents the number of clients, and $c$ represents the parameters of the cache model that need to be trained, i.e., the visual feature weights $W_1$. The local objective $F_k(\cdot)$ is defined by
\begin{align}
    F_k(c)\doteq \frac{1}{n_k}\sum_{j=1}^{n_k}\mathcal{L}(c;(x_{k,j}, y_{k,j})),
\end{align}
where the $k$-th device possesses $n_k$ training data: $(x_{k,1}, y_{k,1}),\cdots,(x_{k,n_k}, y_{k,n_k})$, and $\mathcal{L}(\cdot)$ denotes the cross-entropy loss function.

In the context of a detailed federated learning process, during the $t$-th round, the procedure unfolds as follows: The server initially selects a subset of $K$ clients at random, denoted as $\mathcal{S}_t$. Following this, the server broadcasts the parameters $c_t$ to the participating clients. Each client $k$ then initializes its local parameters as $c_t^k$ = $c_t$. Thereafter, client $k$ proceeds to conduct $E$ rounds of local updates ($i=0, \cdots, E-1$):
\begin{align}
    c_{t+i+1}^k\leftarrow c_{t+i}^k-\eta_{t+i}\nabla F_k(c_{t+i}^k, \xi_{t+i}^k),
\end{align}
where $\eta_{t+i}$ is the learning rate and the $\xi_{t+i}^k$ is sampled from the local dataset at client $k$. The server finally aggregates the received local cache models $\{c_{t+E}^k|k\in \mathcal{S}_t\}$ and updates the global model $\{c_{t+E}\}$ accordingly.

Let $\mathcal{I}_E=\{nE|n=1,2,\cdots \}$ be the set of synchronization rounds for the server-side. We can derive the update rule:
\begin{align}
    v_{t+1}^k & = c_t^k-\eta_t \nabla F_k(c_t^k, \xi_t^k), \\
    c_{t+1}^k & = \begin{cases}
        v_{t+1}^k & \text{if } t+1\notin \mathcal{I}_E, \\
        \sum_{k\in \mathcal{S}_{t+1}} p_k\frac{N}{K}v_{t+1}^k & \text{if } t+1\in \mathcal{I}_E.\\
    \end{cases}
\end{align}

The $v_{t+1}^k$ represents the result of one step SGD, which is practically inaccessible. Also, we define $\bar{g}_t=\sum_{k=1}^N p_k\nabla F_k(c_t^k)$ and $g_t=\sum_{k=1}^N p_k\nabla F_k(c_t^k, \xi_t^k)$. As a result, $\bar{v}_{t+1}=\bar{c}_t-\eta_t g_t$ and $\mathbb{E}g_t=\bar{g}_t$.

\subsection{Assumption and Lemmas}
We begin by stating the common assumptions in federated learning, along with several essential lemmas that will be used in the subsequent proofs.
\begin{assumption}
\label{assumption:1}
    $F_i(i=1, \cdots, N)$ are all L-smooth, then, for all $v$ and $w$, $F_i(v)\leq F_i(w)+(v-w)^\top \nabla F_i(w)+\frac{L}{2}\Vert v-w\Vert_2^2$.
\end{assumption}
\begin{assumption}
\label{assumption:2}
    $F_i(i=1, \cdots, N)$ are all $\mu$-strongly convex, then, for all $v$ and $w$, $F_i(v)\geq F_i(w)+(v-w)^\top \nabla F_i(w)+\frac{L}{2}\Vert v-w \Vert_2^2$.
\end{assumption}
\begin{assumption}
\label{assumption:3}
    The variance of local gradients is bounded: $\mathbb{E}\Vert \nabla F_k(c_t^k, \xi_t^k) - \nabla F_k(c_t^k)\Vert^2\leq \sigma_k^2$.
\end{assumption}
\begin{assumption}
\label{assumption:4}
    The expected squared norm of local gradients is bounded: $\mathbb{E}\Vert \nabla F_k(c_t^k, \xi_t^k)\Vert^2\leq G^2$.
\end{assumption}
Assumptions \ref{assumption:1} and \ref{assumption:2} represent canonical assumptions that are ubiquitously employed within the L2-normative paradigm of normalized linear regression coupled with softmax classification. Conversely, Assumptions \ref{assumption:3} and \ref{assumption:4} are widely adopted in the context of federated learning frameworks.

\begin{lemma}[Results of one step SGD]
    \label{lemma:1}
    When Assumption \ref{assumption:1} and \ref{assumption:2} holds. If $\eta_t\leq \frac{1}{4L}$, we have
    \begin{align}
        \mathbb{E}\Vert \bar{v}_{t+1} - c^*\Vert^2 & \leq (1-\eta_t \mu) \mathbb{E}\Vert \bar{c}_{t} - c^*\Vert^2 \notag \\
        & \quad+ \eta_t^2 \mathbb{E}\Vert g_t - \bar{g}_t\Vert^2 + 6L\eta_t^2 \Gamma \notag \\
        & \quad+ 2\mathbb{E}\sum_{k=1}^{N} p_k\Vert\bar{c}_t - c_k^*\Vert^2,
    \end{align}
    where $\Gamma = F^* - \sum_{k=1}^{N} p_k F_k^* \geq 0$.
\end{lemma}
\begin{proof}
    Considering that the $\bar{v}_{t+1}=\bar{c}_t-\eta_t g_t$ holds,
    \begin{align}
        \Vert \bar{v}_{t+1} - c^*\Vert^2 & = \Vert \bar{c}_t-\eta_t g_t - c^* - \eta_t\bar{g}_t + \eta_t\bar{g}_t\Vert^2 \\
        & = \Vert \bar{c}_t - c^* - \eta_t\bar{g}_t\Vert^2 + \eta_t^2\Vert g_t - \bar{g}_t\Vert^2 \notag \\
        &\quad - 2\eta_t\langle \bar{c}_t - c^* - \eta_t\bar{g}_t, g_t - \bar{g}_t\rangle. \label{equ:l1-12}
    \end{align}
    According to $\mathbb{E}g_t=\bar{g}_t$, $\mathbb{E}[\langle \bar{c}_t - c^* - \eta_t\bar{g}_t, g_t - \bar{g}_t\rangle]=0$. On the other side,
    \begin{align}
        \Vert \bar{c}_t - c^* - \eta_t\bar{g}_t\Vert^2 & = \Vert \bar{c}_t - c^*\Vert^2 + \eta_t^2\Vert \bar{g}_t\Vert^2 - 2\eta_t\langle \bar{c}_t - c^*, \bar{g}_t\rangle. \label{equ:l1-13}
    \end{align}
    From the Assumption \ref{assumption:1}, we have
    \begin{align}
        \Vert \nabla F_k(c_t^k)\Vert^2 \leq 2L(F_k(c_t^k) - F_k^*). \label{equ:l1-14}
    \end{align}
    The L2-norm is inherently a convex function, thus,
    \begin{align}
        \eta_t^2\Vert \bar{g}_t\Vert^2 
        & \leq \eta_t^2\sum_{k=1}^N p_k\Vert \nabla F_k(c_t^k)\Vert^2 \\
        & \leq 2L\eta_t^2 \sum_{k=1}^N p_k(F_k(c_t^k) - F_k^*). \label{equ:l1-15}
    \end{align}
    At the same time,
    \begin{align}
        - 2\eta_t\langle \bar{c}_t - c^*, \bar{g}_t\rangle & = - 2\eta_t\sum_{k=1}^N p_k\langle \bar{c}_t - c^*, \nabla F_k(c_t^k)\rangle \\
        & = - 2\eta_t\sum_{k=1}^N p_k\langle \bar{c}_t - c_t^k, \nabla F_k(c_t^k)\rangle \notag \\
        &\quad - 2\eta_t\sum_{k=1}^N p_k\langle c_t^k - c^*, \nabla F_k(c_t^k)\rangle. \label{equ:l1-17}
    \end{align}
    By the Cauchy-Schwarz inequality and the AM-GM inequality, we have
    \begin{align}
        -2\langle \bar{c}_t - c_t^k, \nabla F_k(c_t^k)\rangle \leq \frac{1}{\eta_t}\Vert \bar{c}_t - c_t^k \Vert^2 + \eta_t \Vert\nabla F_k(c_t^k) \Vert^2. \label{equ:l1-18}
    \end{align}
    From the Assumption \ref{assumption:2}, we have
    \begin{align}
        -\langle c_t^k - c^*, \nabla F_k(c_t^k)\rangle \leq -(F_k(c_t^k) - F_k(c^*)) - \frac{\mu}{2}\Vert c_t^k - c^* \Vert^2. \label{equ:l1-19}
    \end{align}
    By substituting Eqn.(\ref{equ:l1-18}) and Eqn.(\ref{equ:l1-19}) into Eqn.(\ref{equ:l1-17}), and subsequently incorporating the resulting expressions together with Eqn.(\ref{equ:l1-15}) into Eqn.(\ref{equ:l1-13}), we obtain
    \begin{align}
        & \Vert \bar{c}_t - c^* - \eta_t\bar{g}_t\Vert^2 \notag \\
        & \leq \Vert \bar{c}_t - c^*\Vert^2 + 2L\eta_t^2 \sum_{k=1}^N p_k(F_k(c_t^k) - F_k^*) \notag \\
        & \quad + \eta_t\sum_{k=1}^N p_k\left(\frac{1}{\eta_t}\Vert \bar{c}_t - c_t^k \Vert^2 + \eta_t \Vert\nabla F_k(c_t^k) \Vert^2 \right) \notag \\
        & \quad - 2\eta_t\sum_{k=1}^N p_k\left(F_k(c_t^k) - F_k(c^*) + \frac{\mu}{2}\Vert c_t^k - c^* \Vert^2 \right) \label{equ:l1-20}\\
        & = (1-\mu\eta_t)\Vert c_t^k - c^* \Vert^2 + \sum_{k=1}^N p_k\Vert \bar{c}_t - c_t^k \Vert^2 \notag \\
        & \quad + 4L\eta_t^2\sum_{k=1}^N p_k(F_k(c_t^k) - F_k^*) \notag \\
        & \quad - 2\eta_t\sum_{k=1}^N p_k(F_k(c_t^k) - F_k(c^*)).\label{equ:l1-21}
    \end{align}
    Eqn.(\ref{equ:l1-20}) to Eqn.(\ref{equ:l1-21}) are derived from Eqn.(\ref{equ:l1-14}).
    By the notation $\Gamma = F^* - \sum_{k=1}^{N} p_k F_k^* \geq 0$, we have
    \begin{align}
        & \Vert \bar{c}_t - c^* - \eta_t\bar{g}_t\Vert^2 \notag \\
        & \leq (1-\mu\eta_t)\Vert c_t^k - c^* \Vert^2 + \sum_{k=1}^N p_k\Vert \bar{c}_t - c_t^k \Vert^2 \notag \\
        & \quad -2\eta_t(1-2L\eta_t)\sum_{k=1}^N p_k(F_k(c_t^k) - F_k^*) \notag \\
        & \quad + 2\eta_t\sum_{k=1}^N p_k(F_k(c^*) - F_k^*)\\
        & = (1-\mu\eta_t)\Vert c_t^k - c^* \Vert^2 + \sum_{k=1}^N p_k\Vert \bar{c}_t - c_t^k \Vert^2
        \notag \\
        & \quad - 2\eta_t(1-2L\eta_t) \sum_{k=1}^N p_k(F_k(c_t^k) - F^*) \notag \\
        & \quad + 4L\eta_t^2\Gamma. \label{equ:l1-23}
    \end{align}
    By the convexity of $F_k(\cdot)$, the AM-GM inequality and Eqn.(\ref{equ:l1-14}), we obtain
    \begin{align}
        & \sum_{k=1}^N p_k(F_k(c_t^k) - F^*) \notag \\
        & = \sum_{k=1}^N p_k(F_k(c_t^k) - F_k(\bar{c}_t)) + \sum_{k=1}^N p_k(F_k(\bar{c}_t) - F^*) \\
        & \geq \sum_{k=1}^N p_k \langle \nabla F_k(\bar{c}_t), \bar{c}_t^k-\bar{c}_t\rangle + F(\bar{c}_t) - F^*\\
        & \geq -\frac{1}{2} \sum_{k=1}^N p_k\left[\eta_t\Vert\nabla F_k(\bar{c}_t)\Vert^2 + \frac{1}{\eta_t}\Vert\nabla c_t^k - \bar{c}_t\Vert^2 \right] \notag \\
        & \quad + F(\bar{c}_t) - F^*\\
        & \geq - \sum_{k=1}^N p_k\left[\eta_t L(F_k(\bar{c}_t) - F_k^*) + \frac{1}{2\eta_t}\Vert c_t^k - \bar{c}_t\Vert^2\right] 
        \notag \\
        & \quad + F(\bar{c}_t) - F^*.
    \end{align}
    Taking into account the following facts: $\eta_t L-1\leq -\frac{3}{4}<0$ and $\sum_{k=1}^N p_k(F_k(\bar{c}_t) - F^*)=F(\bar{c}_t) - F^*\geq 0$, $\Gamma\geq 0$ and $4L\eta^2+2\eta_t(1-2L\eta_t)\eta_t L = 6L\eta_t^2-4L^2\eta_t^3\leq 6L\eta_t^2$, we have
    \begin{align}
        & - 2\eta_t(1-2L\eta_t) \sum_{k=1}^N p_k(F_k(c_t^k) - F^*) + 4L\eta_t^2\Gamma \notag \\
        & \leq 2\eta_t(1-2L\eta_t) \sum_{k=1}^N p_k\left[\eta_t L(F_k(\bar{c}_t) - F_k^*) + \frac{1}{2\eta_t}\Vert c_t^k - \bar{c}_t\Vert^2\right] \notag \\
        &\quad - 2\eta_t(1-2L\eta_t)(F(\bar{c}_t) - F^*) + 4L\eta_t^2\Gamma\\
        & = 2\eta_t(1-2L\eta_t)(\eta_t L-1)\sum_{k=1}^N p_k(F_k(\bar{c}_t) - F^*) \notag \\
        &\quad + (4L\eta_t^2 + 2\eta_t(1-2L\eta_t)\eta_t L)\Gamma \notag \\
        &\quad + (1-2L\eta_t)\sum_{k=1}^N p_k \Vert c_t^k - \bar{c}_t\Vert^2\\
        & \leq 6L\eta_t^2\Gamma + \sum_{k=1}^N p_k\Vert c_t^k - \bar{c}_t\Vert^2. \label{equ:l1-30}
    \end{align}
    Substituting Eqn.(\ref{equ:l1-30}) into Eqn.(\ref{equ:l1-23}), we have
    \begin{align}
        &\Vert \bar{c}_t - c^* - \eta_t\bar{g}_t\Vert^2 \notag \\ 
        &\leq (1-\mu\eta_t)\Vert c_t^k - c^* \Vert^2 + 2\sum_{k=1}^N p_k\Vert \bar{c}_t - c_t^k \Vert^2 + 6L\eta_t^2\Gamma. \label{equ:l1-31}
    \end{align}
    Using Eqn.(\ref{equ:l1-31}) and taking expectation on both sides of Eqn.(\ref{equ:l1-12}), we have completed the proof of Lemma \ref{lemma:1}.
\end{proof}

\begin{lemma}[Bounding the variance]
    \label{lemma:2}
    When Assumption \ref{assumption:3} holds. we have
    \begin{align}
        \mathbb{E}\Vert g_t-\bar{g}_t \Vert^2 \leq \sum_{k=1}^N p_k^2\sigma^2.
    \end{align}
\end{lemma}
\begin{proof}
    Because the variance of local gradients is bounded, then
    \begin{align}
        \mathbb{E}\Vert g_t-\bar{g}_t \Vert^2 & = \mathbb{E}\left\| \sum_{k=1}^N p_k(\nabla F_k(c_t^k, \xi_t^k) - \nabla F_k(c_t^k)) \right\|^2 \\
        & = \sum_{k=1}^N p_k^2\mathbb{E}\left\| \nabla F_k(c_t^k, \xi_t^k) - \nabla F_k(c_t^k) \right\|^2 \\
        & \leq \sum_{k=1}^N p_k^2\sigma^2.
    \end{align}
\end{proof}

\begin{lemma}[Bounding the divergence of $\{c_t^k\}$]
    \label{lemma:3}
    When Assumption \ref{assumption:4} holds. If $\eta_t$ is non-increasing and $\eta_t\leq 2\eta_{t+E}$ for all $t\geq 0$, we have
    \begin{align}
        \mathbb{E}\left[\sum_{k=1}^N p_k\left\| \bar{c}_t-c_t^k \right\|^2 \right] \leq 4\eta_2(E-1)^2G^2.
    \end{align}
\end{lemma}
\begin{proof}
    For any $t\geq 0$, there exists a $t_0\leq t$, s.t. $t-t_0\leq E-1$ and $c_{t_0}^k=\bar{c}_{t_0}$ for all $k=1, 2, \cdots, N$.
    \begin{align}
        & \mathbb{E}\left[\sum_{k=1}^N p_k\left\| \bar{c}_t-c_t^k \right\|^2 \right] \notag \\
        & = \mathbb{E}\left[\sum_{k=1}^N p_k\left\| (c_t^k-\bar{c}_{t_0}) - (\bar{c}_t-\bar{c}_{t_0}) \right\|^2 \right] \\
        & \leq \mathbb{E}\left[\sum_{k=1}^N p_k\left\| c_t^k-\bar{c}_{t_0} \right\|^2 \right] \\
        & \leq \sum_{k=1}^N p_k\mathbb{E}\left[\sum_{t=t_0}^{t-1}(E-1)\eta_t^2\left\|\nabla F_k(c_t^k, \xi_t^k) \right\|^2 \right] \\
        & \leq \sum_{k=1}^N p_k\sum_{t=t_0}^{t-1}(E-1)\eta_{t_0}^2 G^2 \\
        & \leq \sum_{k=1}^N p_k\eta_{t_0}^2(E-1)^2G^2 \\
        & \leq 4\eta_t^2(E-1)^2G^2.
    \end{align}
\end{proof}

\begin{lemma}[Unbiased sampling scheme]
    \label{lemma:4}
    If $t+1\in\mathcal{I}_E$, we have
    \begin{align}
        \mathbb{E}_{\mathcal{S}_t}(\bar{c}_{t+1}) = \bar{v}_{t+1}.
    \end{align}
\end{lemma}
\begin{proof}
    We sample $x_k$ with probability $q_k=\frac{1}{N}$ for each time. Let $\mathcal{S}_t = \{i_1, \cdots, i_K\}\subset [N]$. Then
    \begin{align}
    \mathbb{E}_{\mathcal{S}_t}\sum_{k\in\mathcal{S}_t}x_k & = \mathbb{E}_{\mathcal{S}_t}\sum_{k=1}^K x_{i_k} = K\mathbb{E}_{\mathcal{S}_t}x_{i_1} 
    \notag \\
    & = K\sum_{k=1}^N q_k x_k = \frac{K}{N}\sum_{k=1}^N x_k.
    \end{align}
\end{proof}

\begin{lemma}[Bounding the variance of $\bar{c}_t$]
    \label{lemma:5}
    If $\eta_t$ is non-increasing and $\eta_t\leq 2\eta_{t+E}$ for all $t\geq 0$, the expected difference between $\bar{v}_{t+1}$ and $\bar{c}_{t+1}$ is bounded by
    \begin{align}
        \mathbb{E}_{\mathcal{S}_t}\left\|\bar{v}_{t+1} - \bar{c}_{t+1}\right\|^2\leq \frac{N-K}{N-1}\frac{4}{K}\eta_t^2E^2G^2.
    \end{align}
\end{lemma}
\begin{proof}
    Suppose $p_i=\frac{1}{N}, i=1, 2, \cdots, N$. We have $\bar{c}_{t+1} = \frac{1}{K}\sum_{k=1}^K v_{t+1}^{i_k}$.
    \begin{align}
        &\mathbb{E}_{\mathcal{S}_t}\left\|\bar{v}_{t+1} - \bar{c}_{t+1}\right\|^2 \notag \\
        & = \mathbb{E}_{\mathcal{S}_t}\left\|\frac{1}{K}\sum_{i\in \mathcal{S}_{t+1}}v^i_{t+1} - \bar{v}_{t+1}\right\|^2 \\
        & = \frac{1}{K^2}\mathbb{E}_{\mathcal{S}_t}\left\|\sum_{i=1}^N\mathbb{I}\{i\in \mathcal{S}_t\} (v_{t+1}^i-\bar{v}_{t+1})\right\|^2 \\
        % & = \frac{1}{K^2}\left[\sum_{i\in[N]}\mathbb{P}(i\in \mathcal{S}_{t+1})\left\|v_{t+1}^i - \bar{v}_{t+1} \right\|^2 + \sum_{i\neq j}\mathbb{P}(i, j\in \mathcal{S}_{t+1})\langle v_{t+1}^i - \bar{v}_{t+1}, v_{t+1}^j - \bar{v}_{t+1}\rangle  \right] \\
        & = \frac{1}{K^2}\sum_{i\in[N]}\mathbb{P}(i\in \mathcal{S}_{t+1})\left\|v_{t+1}^i - \bar{v}_{t+1} \right\|^2 \notag \\
        & \quad + \frac{1}{K^2}\sum_{i\neq j}\mathbb{P}(i, j\in \mathcal{S}_{t+1})\langle v_{t+1}^i - \bar{v}_{t+1}, v_{t+1}^j - \bar{v}_{t+1}\rangle \\
        & = \frac{1}{KN}\sum_{i=1}^N\left\|v_{t+1}^i - \bar{v}_{t+1} \right\|^2 \notag \\
        & \quad + \sum_{i\neq j}\frac{K-1}{KN(N-1)}\langle v_{t+1}^i - \bar{v}_{t+1}, v_{t+1}^j - \bar{v}_{t+1}\rangle \\
        & = \frac{1}{K(N-1)}\left(1-\frac{K}{N}\right)\sum_{i=1}^N\left\|v_{t+1}^i - \bar{v}_{t+1} \right\|^2,
    \end{align}
    which is using the following equalities: $\mathbb{P}(i\in \mathcal{S}_{t+1}) = \frac{K}{N}$ and $\mathbb{P}(i, j\in \mathcal{S}_{t+1}) = \frac{K(K-1)}{N(N-1)}$ for all $i\neq j$, and $\sum_{i\in[N]}\left\|v_{t+1}^i - \bar{v}_{t+1} \right\|^2 + \sum_{i\neq j}\langle v_{t+1}^i - \bar{v}_{t+1}, v_{t+1}^j - \bar{v}_{t+1}\rangle = 0$.
    \begin{align}
        & \mathbb{E}\left\|\bar{v}_{t+1} - \bar{c}_{t+1}\right\|^2 \notag \\
        & = \frac{N}{K(N-1)}\left(1-\frac{K}{N}\right)\mathbb{E}\left[\frac{1}{N}\sum_{i=1}^N\left\|v_{t+1}^i - \bar{v}_{t+1} \right\|^2 \right]\\
        & \leq \frac{N}{K(N-1)}\left(1-\frac{K}{N}\right)\mathbb{E}\left[\frac{1}{N}\sum_{i=1}^N\left\|v_{t+1}^i - \bar{c}_{t_0} \right\|^2 \right]\\
        & \leq \frac{N}{K(N-1)}\left(1-\frac{K}{N}\right)4\eta_t^2E^2G^2.
    \end{align}
\end{proof}

\subsection{The Convergence of CacheFL}
Theorem \ref{the1} demonstrates the convergence of CacheFL.
\begin{theorem}
    \label{the1}
    Suppose Assumptions 1 to 4 hold. Choose $\gamma = \max\{8\frac{L}{\mu}-1\}$ and the learning rate $\eta_t = \frac{2}{\mu}\frac{1}{t+\gamma}$. Then the objective function of CacheFL satisfies
    \begin{align}
        \mathbb{E}[F(\bar{c}_t)] - F^*\leq \frac{2L}{(t+\gamma)\mu}\left(\frac{B+C}{\mu} + 2L\left\|\bar{c}_{1} - c^* \right\|^2\right),
    \end{align}
    where
    \begin{align}
        B & = \sum_{k=1}^N p_k^2\sigma^2 + 6L\Gamma + 8(E-1)^2G^2,\\
        C & = \frac{N-K}{N-1}\frac{4}{K}E^2G^2.
    \end{align}
\end{theorem}

\begin{proof}
    Note that
    \begin{align}
        \left\|\bar{c}_{t+1} - c^* \right\|^2 & = \left\|\bar{c}_{t+1} - \bar{v}_{t+1} + \bar{v}_{t+1} - c^* \right\|^2 \\
        & = \left\|\bar{c}_{t+1} - \bar{v}_{t+1} \right\|^2 + \left\|\bar{v}_{t+1} - c^* \right\|^2 \notag \\
        & \quad + 2\langle \bar{c}_{t+1} - \bar{v}_{t+1}, \bar{v}_{t+1} - c^*\rangle.
    \end{align}
    Taking expectation over $\mathcal{S}_{t+1}$, the inner product is zero due to the unbiasedness of $\bar{c}_{t+1}$, which has been proved in Lemma \ref{lemma:4}.
    
    If $t+1\notin \mathcal{I}_E$, $A_1$ vanishes because of $\bar{c}_{t+1} = \bar{v}_{t+1}$. We use Lemma \ref{lemma:1}, Lemma \ref{lemma:2}, Lemma \ref{lemma:3} to bound the second terms:
    \begin{align}
        \mathbb{E}\left\|\bar{c}_{t+1} - c^* \right\|^2 & = \mathbb{E}\left\|\bar{v}_{t+1} - c^* \right\|^2 \\
        & \leq (1-\eta_t\mu)\mathbb{E}\left\|\bar{c}_{t} - c^* \right\|^2 + \eta_t^2B,
    \end{align}
    where 
    \begin{align}
        B = \sum_{k=1}^N p_k^2\sigma^2 + 6L\Gamma + 8(E-1)^2G^2.
    \end{align}

    If $t+1\in \mathcal{I}_E$, we additionally use Lemma \ref{lemma:5} to bound the first terms:
    \begin{align}
        \mathbb{E}\left\|\bar{c}_{t+1} - c^* \right\|^2 & = \mathbb{E}\left\|\bar{c}_{t+1} - \bar{v}_{t+1} \right\|^2 + \mathbb{E}\left\|\bar{v}_{t+1} - c^* \right\|^2 \\
        & \leq (1-\eta_t\mu)\mathbb{E}\left\|\bar{c}_{t} - c^* \right\|^2 + \eta_t^2(B + C),
    \end{align}
    where
    \begin{align}
        C = \frac{N-K}{N-1}\frac{4}{K}E^2G^2.
    \end{align}
    Let $\Delta_t = \mathbb{E}\left\|\bar{c}_{t} - c^* \right\|^2$, we have
    \begin{align}
        \Delta_{t+1}\leq (1-\eta_t\mu)\Delta_t + \eta_t^2(B + C).
    \end{align}
    For decreasing learning rates, $\eta_t = \frac{\beta}{t+\gamma}$ for some $\beta>\frac{1}{\mu}$ and $\gamma>0$ s.t. $\eta_1\leq \min\{\frac{1}{\mu}, \frac{1}{4L}\}$ and $\eta_t\leq 2\eta_{t+E}$. We're going to prove $\Delta_t\leq \frac{v}{t+\gamma}$ where $v=\max\{\frac{\beta^2(B+C)}{\beta\mu - 1}, (\gamma + 1)\Delta_1\}$.

    We use mathematical induction to prove that. When $t=1$, $\Delta_1\leq \max\{\frac{\beta^2(B+C)}{\beta\mu - 1}, (\gamma + 1)\Delta_1\}$ clearly holds. Assuming that $t$ holds for the inequality, we have
    \begin{align}
        \Delta_{t+1} & \leq (1-\eta_t\mu)\Delta_t + \eta_t^2(B + C)\\
        & = \left(1 - \frac{\beta\mu}{t+\gamma}\right)\frac{v}{t+\gamma} + \frac{\beta^2(B+C)}{(t+\gamma)^2}\\
        & = \frac{t+\gamma - \beta\mu}{(t+\gamma)^2}v + \frac{\beta^2(B+C)}{(t+\gamma)^2}\\
        & = \frac{t+\gamma - 1}{(t+\gamma)^2}v + \left[\frac{\beta^2(B+C)}{(t+\gamma)^2} - \frac{\beta\mu - 1}{(t+\gamma)^2}v \right].
    \end{align}
    According to $v\geq\frac{\beta^2(B+C)}{\beta\mu - 1}$, we have
    \begin{align}
        \Delta_{t+1} & \leq \frac{t+\gamma - 1}{(t+\gamma)^2}v \\
        & = \frac{(t+\gamma - 1)(t+\gamma+1)}{(t+\gamma)^2(t+\gamma+1)}v \\
        & = \frac{(t+\gamma)^2-1}{(t+\gamma)^2(t+\gamma+1)}v \\
        & \leq \frac{v}{t+\gamma+1} \leq \frac{v}{t+\gamma}.
    \end{align}
    Then by the Assumption \ref{assumption:1}, we have
    \begin{align}
        \mathbb{E}[F(\bar{c}_t)] - F^*\leq \frac{L}{2}\Delta_t\leq \frac{L}{2}\frac{v}{t+\gamma}.
    \end{align}
    Specifically, if we choose $\beta=\frac{2}{\mu}$, $\gamma = \max\{8\frac{L}{\mu}-1\}$, then $\eta_t = \frac{2}{\mu}\frac{1}{t+\gamma}$ and
    \begin{align}
        \mathbb{E}[F(\bar{c}_t)] - F^*\leq \frac{2L}{(t+\gamma)\mu}\left(\frac{B+C}{\mu} + 2L\left\|\bar{c}_{1} - c^* \right\|^2\right).
    \end{align}
\end{proof}

\section{Experiments}\label{4}
In this section, we evaluate CacheFL on several public image classification datasets, comparing its performance and convergence speed with other methods. Additionally, we assess CacheFL's resource consumption and privacy-preserving capabilities relative to alternative approaches. Furthermore, we examine the impact of using synthetic dataset initialization, federated training, and different hyperparameters on the overall effectiveness of the method.

\subsection{Experimental Settings}
\noindent \textbf{Datasets and Data Heterogeneity.} 
We conduct experiments for CacheFL on 11 widely used image classification datasets: ImageNet \cite{deng2009imagenet}, Caltech101 \cite{fei2004learning}, DTD \cite{cimpoi2014describing}, EuroSAT \cite{helber2019eurosat}, FGVCAircraft \cite{maji2013fine}, Food101 \cite{bossard2014food},  Flowers102 \cite{nilsback2008automated}, OxfordPets \cite{parkhi2012cats}, StandfordCars \cite{krause20133d}, SUN397 \cite{xiao2010sun}, and UCF101 \cite{soomro2012ucf101}.
The raw training data for all 11 datasets consists of 16 images per class. To evaluate the robustness of our approach against data heterogeneity, we adopt three distinct methods of dataset partitioning:
\begin{itemize}
    \item \textbf{IID} (iid): The data is evenly distributed across all clients. Each client possesses an identical set of classes, with a relatively uniform distribution in quantity. For instance, in the case where 16 images are allocated to 10 clients, the first four clients are provided with one image per class, whereas the subsequent six clients are provided with two images per class.
    \item \textbf{Dirichlet non-IID} (dir): The data is partitioned randomly among clients using a Dirichlet distribution with $Dir(0.1)$.
    \item \textbf{Extreme non-IID} (pat): Each client owns the independent and non-overlapping classes. For instance, when distributing 100 classes across 10 clients, each client receives a distinct, non-overlapping set of 10 classes.
\end{itemize}

\noindent \textbf{Baselines methods.}
We perform a performance comparison between six approaches: FedAvg \cite{mcmahan2017communication}, FedProx \cite{li2020federated}, CLIP \cite{radford2021learning}, AdapterFL (integrate adapter tuning with federated learning), PromptFL \cite{guo2023promptfl}, and CacheFL. 
In FedAvg and FedProx, we use ResNet-50 \cite{he2016deep} as the backbone for training from scratch. For methods based on CLIP, we adopt ResNet-50 as the image encoder and a transformer as the text encoder. 

\begin{itemize}
    \item \textbf{FedAvg}: A widely used federated learning framework. ResNet-50 is employed as the backbone model, trained from scratch for 500 rounds, with each round consisting of one epoch. 
    \item \textbf{FedProx}: An enhanced optimization method based on FedAvg, which introduces a proximal term to address challenges posed by non-IID data distributions. All other settings are the same as FedAvg, with the proximal term parameter $\mu = 0.1$.
    \item \textbf{CLIP}: The original CLIP model without any fine-tuning.
    \item \textbf{AdapterFL}: An adapter-based FL method for CLIP fine-tuning, which learns an adapter across clients. The adapter consists of two fully connected layers and two ReLU layers, integrated into the image encoder. The original CLIP model weights remain frozen, and only the adapter layers are updated during training. The global training includes 20 rounds, with one local epoch per round.
    \item \textbf{PromptFL}: A prompt-based FL method for CLIP fine-tuning, which learns a prompt across clients. The prompt consists of $p$ learnable $d$-dimensional vectors, where $d$ corresponds to the dimension of word embeddings in the text encoder (default is 512), and $p$ is a hyperparameter (default is 16). Similar to AdapterFL, the global training consists of 20 rounds, with one local epoch per round.
    \item \textbf{CacheFL}:  This method adopts DALL·E to generate 16 synthetic images for each category. The cache model is initialized with this synthetic dataset, and training follows Algorithm \ref{fedavg1} with $T = 20$ (global rounds) and $E = 1$ (local epochs). 
    % For the two hyperparameters $\alpha$ and $\beta$, we use the grid search to adjust the optimal values for each dataset.
\end{itemize}

\begin{table*}[htbp]
    \centering
    \caption{Performance of CacheFL Against Existing FL Method. }
    \label{tab:results}
    \begin{minipage}[t]{0.24\textwidth}
        \centering
        \textbf{\scriptsize (a) Average over 11 datasets} \\[0.5ex]
        % \caption*{\textbf{Average over 11 datasets}}
        \begin{tabularx}{\textwidth}{lYYY}
            \toprule
            Method           & iid           & dir           & pat           \\ 
            \midrule
            FedAvg           & 5.2           & 5.1           & 4.4           \\
            FedProx          & 10.3           & 9.8           & 7.3           \\
            CLIP             & 60.8          & 60.8          & 60.8          \\
            AdapterFL        & 56.1          & 56.2          & 56.2          \\
            PromptFL         & 64.0          & 63.1          & 62.3          \\
            \rowcolor[gray]{0.9} CacheFL & \textbf{73.6} & \textbf{70.9} & \textbf{69.4} \\ 
            & \textcolor{purple}{+9.6} & \textcolor{purple}{+7.8} & \textcolor{purple}{+7.1} \\ 
            \bottomrule
        \end{tabularx}
    \end{minipage}
    \hfill
    \begin{minipage}[t]{0.24\textwidth}
        \centering
        \textbf{\scriptsize (b) ImageNet} \\[0.5ex]
        % \caption*{ImageNet}
        \begin{tabularx}{\textwidth}{lYYY}
            \toprule
            Method           & iid           & dir           & pat           \\ 
            \midrule
            FedAvg           & 0.4           & 0.3           & 0.2           \\
            FedProx          & 0.4          & 0.3          & 0.3          \\
            CLIP             & 61.5          & 61.5          & 61.5          \\
            AdapterFL        & 58.4          & 58.3          & 58.3          \\
            PromptFL         & 62.1          & 61.8          & 62.1          \\
            \rowcolor[gray]{0.9} CacheFL & \textbf{64.3} & \textbf{64.0} & \textbf{64.0} \\ 
            & \textcolor{purple}{+2.2} & \textcolor{purple}{+2.2} & \textcolor{purple}{+1.9} \\ 
            \bottomrule
        \end{tabularx}
    \end{minipage} 
    \hfill
    \begin{minipage}[t]{0.24\textwidth}
        \centering
        \textbf{\scriptsize (c) Caltech101} \\[0.5ex]
        % \caption*{Caltech101}
        \begin{tabularx}{\textwidth}{lYYY}
            \toprule
            Method           & iid           & dir           & pat           \\ 
            \midrule
            FedAvg           & 6.0           & 9.3           & 5.7           \\
            FedProx          & 19.4          & 15.8          & 20.0          \\
            CLIP             & 88.2          & 88.2          & 88.2          \\
            AdapterFL        & 85.6          & 85.0          & 85.0          \\
            PromptFL         & 90.2          & 90.7          & 88.7          \\
            \rowcolor[gray]{0.9} CacheFL & \textbf{92.8} & \textbf{92.6} & \textbf{92.7} \\ 
            & \textcolor{purple}{+2.6} & \textcolor{purple}{+1.9} & \textcolor{purple}{+4.0} \\ 
            \bottomrule
        \end{tabularx}
    \end{minipage} 
    \hfill 
    \begin{minipage}[t]{0.24\textwidth}
        \centering
        \textbf{\scriptsize (d) DTD} \\[0.5ex]
        % \caption*{DTD}
        \begin{tabularx}{\textwidth}{lYYY}
            \toprule
            Method           & iid           & dir           & pat           \\ 
            \midrule
            FedAvg           & 4.3           & 4.4           & 4.6           \\
            FedProx          & 9.9          & 8.4          & 7.3          \\
            CLIP             & 50.1          & 50.1          & 50.1          \\
            AdapterFL        & 39.1          & 39.4          & 39.3          \\
            PromptFL         & 55.0          & 48.5          & 44.4          \\
            \rowcolor[gray]{0.9} CacheFL & \textbf{66.0} & \textbf{65.2} & \textbf{63.5} \\
            & \textcolor{purple}{+11.0} & \textcolor{purple}{+16.7} & \textcolor{purple}{+19.1} \\ 
            \bottomrule
        \end{tabularx}
    \end{minipage}
    \vspace{1em}
  
    \begin{minipage}[t]{0.24\textwidth}
        \centering
        \textbf{\scriptsize (e) EuroSAT} \\[0.5ex]
        % \caption*{EuroSAT}
        \begin{tabularx}{\textwidth}{lYYY}
            \toprule
            Method           & iid           & dir           & pat           \\ 
            \midrule
            FedAvg           & 26.8          & 24.9          & 20.0          \\
            FedProx          & 42.8          & 49.3          & 21.2          \\
            CLIP             & 37.4          & 37.4          & 37.4          \\
            AdapterFL        & 25.9          & 26.2          & 26.7          \\
            PromptFL         & 45.5          & 52.5          & 54.1          \\
            \rowcolor[gray]{0.9} CacheFL & \textbf{81.6} & \textbf{66.2} & \textbf{63.6} \\
            & \textcolor{purple}{+36.1} & \textcolor{purple}{+13.7} & \textcolor{purple}{+9.5} \\ 
            \bottomrule
        \end{tabularx}
    \end{minipage} 
    \hfill
    \begin{minipage}[t]{0.24\textwidth}
        \centering
        \textbf{\scriptsize (f) FGVCAircraft} \\[0.5ex]
        % \caption*{FGVCAircraft}
        \begin{tabularx}{\textwidth}{lYYY}
            \toprule
            Method           & iid           & dir           & pat           \\ 
            \midrule
            FedAvg           & 2.7           & 2.7           & 2.6           \\
            FedProx          & 6.6          & 5.0          & 4.0          \\
            CLIP             & 20.7          & 20.7          & 20.7          \\
            AdapterFL        & 15.8          & 16.0          & 16.3          \\
            PromptFL         & 6.2           & 19.0          & 15.6          \\
            \rowcolor[gray]{0.9} CacheFL & \textbf{29.3} & \textbf{26.6} & \textbf{25.8} \\
            & \textcolor{purple}{+23.1} & \textcolor{purple}{+7.6} & \textcolor{purple}{+10.2} \\ 
            \bottomrule
        \end{tabularx}
    \end{minipage} 
    \hfill
    \begin{minipage}[t]{0.24\textwidth}
        \centering
        \textbf{\scriptsize (g) Food101} \\[0.5ex]
        % \caption*{Food101}
        \begin{tabularx}{\textwidth}{lYYY}
            \toprule
            Method           & iid           & dir           & pat           \\ 
            \midrule
            FedAvg           & 1.7           & 1.6           & 1.5           \\
            FedProx          & 1.6          & 1.8          & 1.7          \\
            CLIP             & 77.7          & 77.7          & 77.7          \\
            AdapterFL        & 74.5          & 74.3          & 74.3          \\
            PromptFL         & 78.0          & 77.4          & 78.1          \\
            \rowcolor[gray]{0.9} CacheFL & \textbf{78.5} & \textbf{78.5} & \textbf{78.5} \\
            & \textcolor{purple}{+0.5} & \textcolor{purple}{+1.1} & \textcolor{purple}{+0.4} \\ 
            \bottomrule
        \end{tabularx}
    \end{minipage} 
    \hfill
    \begin{minipage}[t]{0.24\textwidth}
        \centering
        \textbf{\scriptsize (h) Flowers102} \\[0.5ex]
        % \caption*{Flowers102}
        \begin{tabularx}{\textwidth}{lYYY}
            \toprule
            Method           & iid           & dir           & pat           \\ 
            \midrule
            FedAvg           & 5.2           & 4.4           & 4.5           \\
            FedProx          & 3.7          & 3.7          & 3.3          \\
            CLIP             & 64.2          & 64.2          & 64.2          \\
            AdapterFL        & 61.1          & 61.1          & 61.1          \\
            PromptFL         & 88.2          & 68.6          & 66.3          \\
            \rowcolor[gray]{0.9} CacheFL & \textbf{94.2} & \textbf{90.9} & \textbf{82.1} \\
            & \textcolor{purple}{+6.0} & \textcolor{purple}{+22.3} & \textcolor{purple}{+15.8} \\ 
            \bottomrule
        \end{tabularx}
    \end{minipage}
    \vspace{1em}
    
    \begin{minipage}[t]{0.24\textwidth}
        \centering
        \textbf{\scriptsize (i) OxfordPets} \\[0.5ex]
        % \caption*{OxfordPets}
        \begin{tabularx}{\textwidth}{lYYY}
            \toprule
            Method           & iid           & dir           & pat           \\ 
            \midrule
            FedAvg           & 5.0           & 4.0           & 4.4           \\
            FedProx          & 8.9          & 5.1          & 8.2          \\
            CLIP             & 86.4          & 86.4          & 86.4          \\
            AdapterFL        & 83.1          & 83.1          & 83.0          \\
            PromptFL         & 88.5          & 86.9          & 87.0          \\
            \rowcolor[gray]{0.9} CacheFL & \textbf{89.5} & \textbf{88.3} & \textbf{88.7} \\
            & \textcolor{purple}{+1.0} & \textcolor{purple}{+1.4} & \textcolor{purple}{+1.7} \\ 
            \bottomrule
        \end{tabularx}
    \end{minipage} 
    \hfill
    \begin{minipage}[t]{0.24\textwidth}
        \centering
        \textbf{\scriptsize (j) StandfordCars} \\[0.5ex]
        % \caption*{StandfordCars}
        \begin{tabularx}{\textwidth}{lYYY}
            \toprule
            Method           & iid           & dir           & pat           \\ 
            \midrule
            FedAvg           & 1.6           & 1.2           & 1.2           \\
            FedProx          & 3.6          & 2.3          & 2.6          \\
            CLIP             & 57.0          & 57.0          & 57.0          \\
            AdapterFL        & 56.1          & 56.2          & 56.0          \\
            PromptFL         & 58.7          & 58.3          & 58.4          \\
            \rowcolor[gray]{0.9} CacheFL & \textbf{66.0} & \textbf{64.5} & \textbf{64.7} \\
            & \textcolor{purple}{+7.3} & \textcolor{purple}{+6.2} & \textcolor{purple}{+6.3} \\ 
            \bottomrule
        \end{tabularx}
    \end{minipage} 
    \hfill
    \begin{minipage}[t]{0.24\textwidth}
        \centering
        \textbf{\scriptsize (k) SUN397} \\[0.5ex]
        % \caption*{SUN397}
        \begin{tabularx}{\textwidth}{lYYY}
            \toprule
            Method           & iid           & dir           & pat           \\ 
            \midrule
            FedAvg           & 0.9           & 0.9           & 0.8           \\
            FedProx          & 0.9          & 0.8          & 0.8          \\
            CLIP             & 62.3          & 62.3          & 62.3          \\
            AdapterFL        & 59.1          & 59.2          & 58.8          \\
            PromptFL         & 64.2          & 63.6          & 64.4          \\
            \rowcolor[gray]{0.9} CacheFL & \textbf{70.1} & \textbf{69.1} & \textbf{68.5} \\ 
            & \textcolor{purple}{+5.9} & \textcolor{purple}{+5.5} & \textcolor{purple}{+4.1} \\ 
            \bottomrule
        \end{tabularx}
    \end{minipage} 
    \hfill
    \begin{minipage}[t]{0.24\textwidth}
        \centering
        \textbf{\scriptsize (l) UCF101} \\[0.5ex]
        % \caption*{UCF101}
        \begin{tabularx}{\textwidth}{lYYY}
            \toprule
            Method           & iid           & dir           & pat           \\ 
            \midrule
            FedAvg           & 3.2           & 2.4           & 2.3           \\
            FedProx          & 15.4          & 15.6          & 10.8          \\
            CLIP             & 63.3          & 63.3          & 63.3          \\
            AdapterFL        & 59.1          & 59.1          & 59.0          \\
            PromptFL         & 67.3          & 66.7          & 66.5          \\
            \rowcolor[gray]{0.9} CacheFL & \textbf{77.6} & \textbf{74.2} & \textbf{71.5} \\
            & \textcolor{purple}{+10.3} & \textcolor{purple}{+7.5} & \textcolor{purple}{+5.0} \\
            \bottomrule
        \end{tabularx}
    \end{minipage} 
    \vspace{2pt}
    \begin{flushleft}We evaluate the performance in terms of accuracy on 11 datasets. We adopt three data partitioning methods, where ``iid" stands for the IID setting, ``dir" represents the Dirichlet non-IID setting, and ``pat" denotes the Extreme non-IID setting. Absolute gains over PromptFL are indicated in \textcolor{purple}{purple}.\end{flushleft}
\end{table*}

\noindent \textbf{Implementation Details.}
All experiments are conducted using PyTorch on Tesla V100 SXM2 32GB GPU, with training performed using stochastic gradient descent (SGD) with a learning rate of 0.001. The loss function is the cross-entropy loss. 
% CacheFL adopts DALL·E to generate 16 synthetic images per category. For the two hyperparameters $\alpha$ and $\beta$, we use the grid search to adjust the optimal values for each dataset. 

\noindent \textbf{Evaluation Metrics.} We use accuracy as the primary metric to evaluate model performance. Resource consumption is assessed by communication cost and computation cost. To evaluate the security and privacy protection of CacheFL, we employ a representative gradient inversion attack \cite{zhu2019deep} to attempt the reconstruction of raw training images, testing the model's resilience to such attacks.

\begin{figure*}[t]
   \centering
   \centering
   \subfloat{\includegraphics[width=0.165\textwidth]{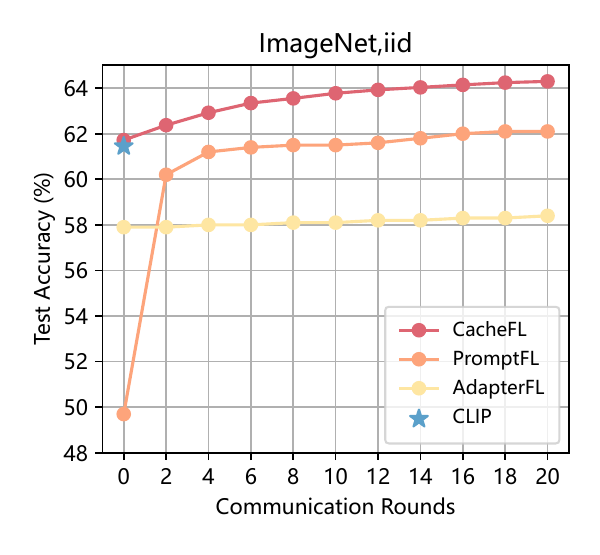}}
   \subfloat{\includegraphics[width=0.165\textwidth]{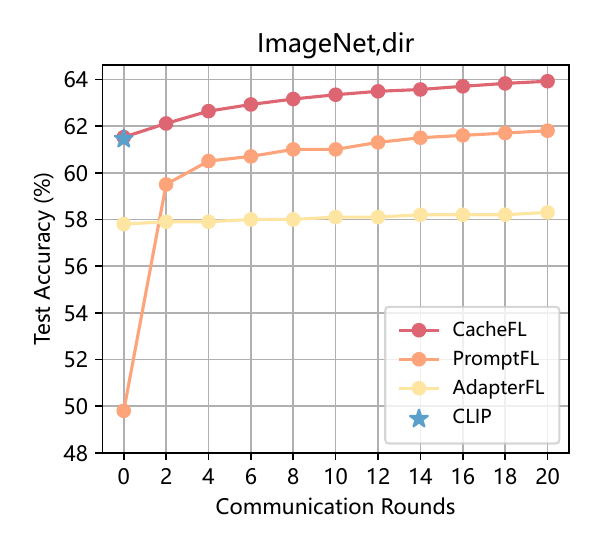}}
   \subfloat{\includegraphics[width=0.165\textwidth]{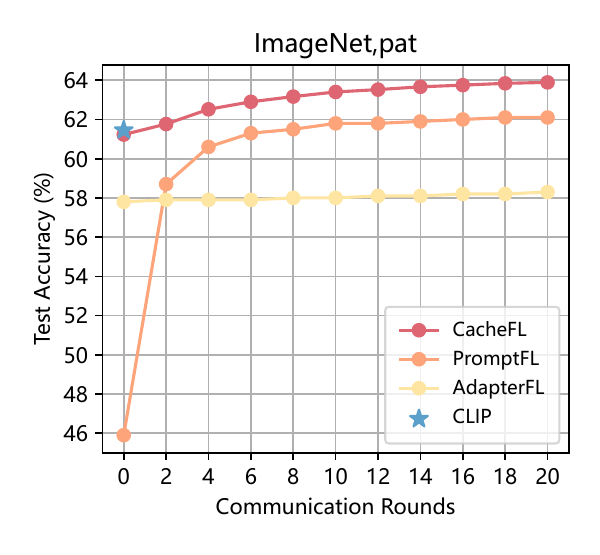}}
   \subfloat{\includegraphics[width=0.165\textwidth]{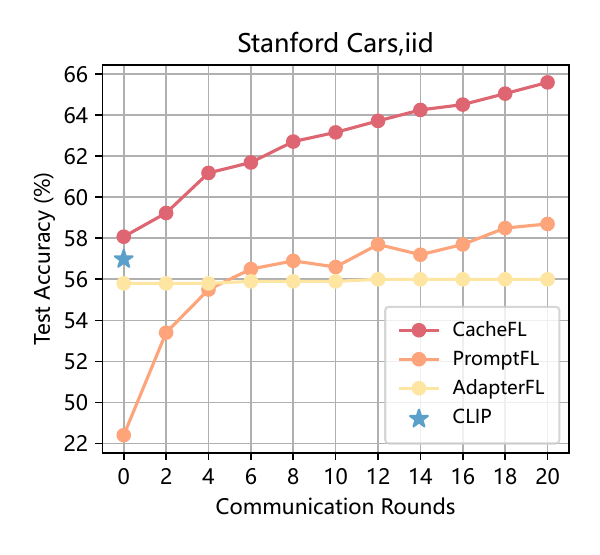}}
   \subfloat{\includegraphics[width=0.165\textwidth]{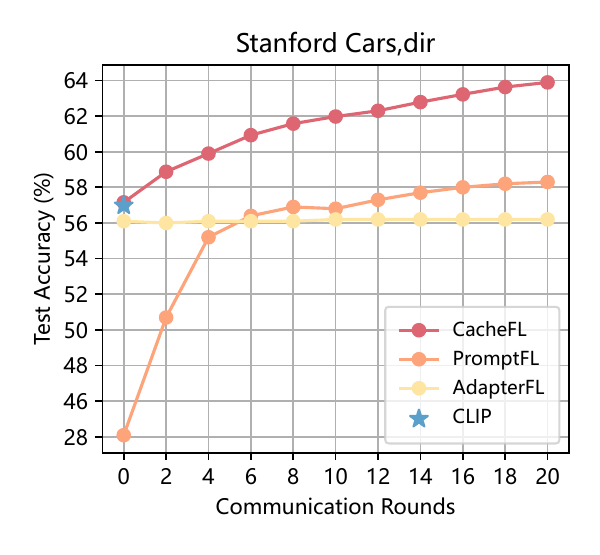}}
   \subfloat{\includegraphics[width=0.165\textwidth]{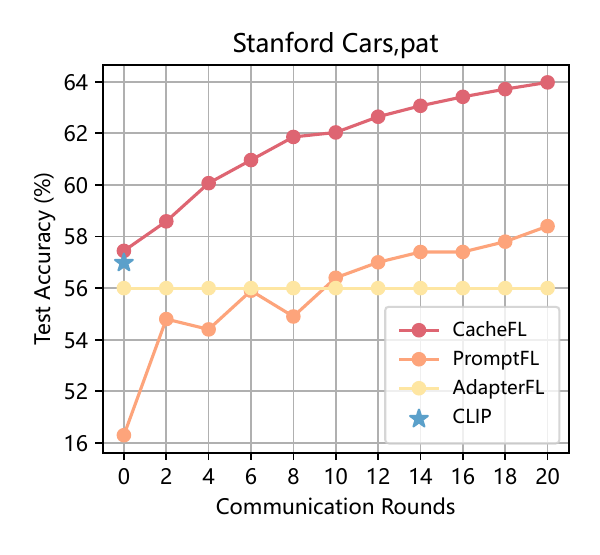}}
   \\
   \subfloat{\includegraphics[width=0.165\textwidth]{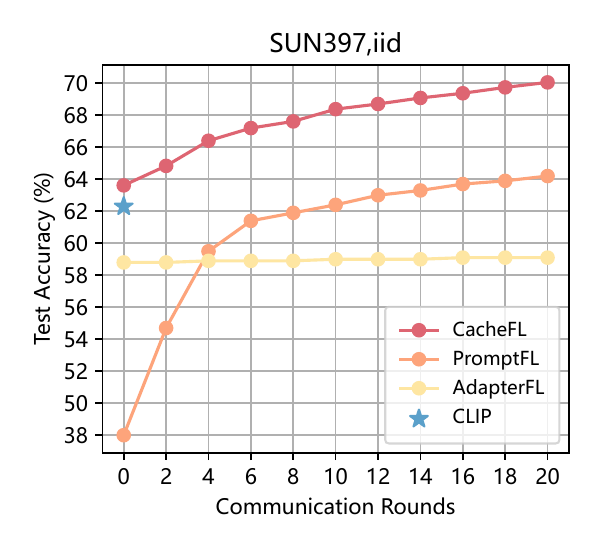}}
   \subfloat{\includegraphics[width=0.165\textwidth]{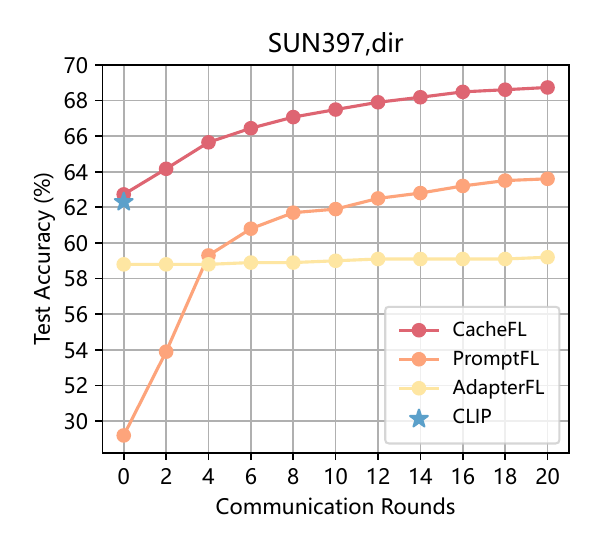}}
   \subfloat{\includegraphics[width=0.165\textwidth]{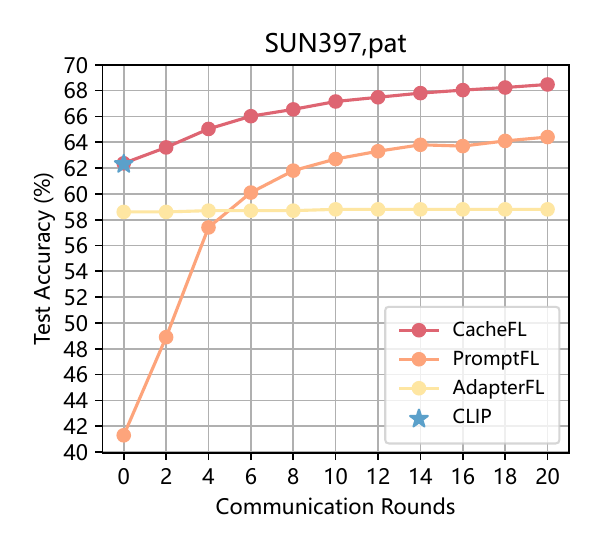}}
   \subfloat{\includegraphics[width=0.165\textwidth]{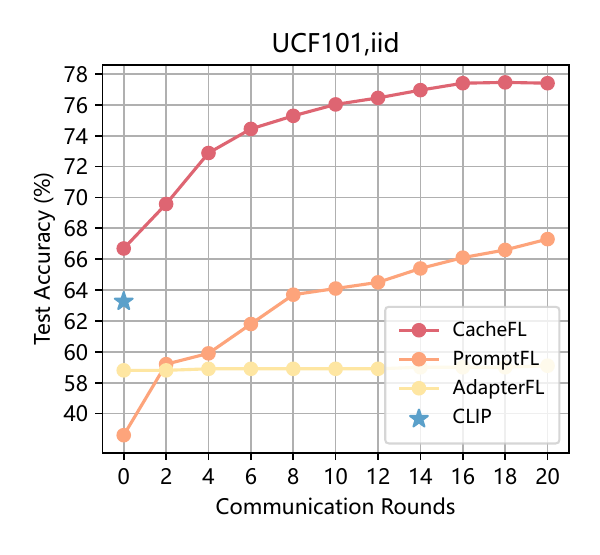}}
   \subfloat{\includegraphics[width=0.165\textwidth]{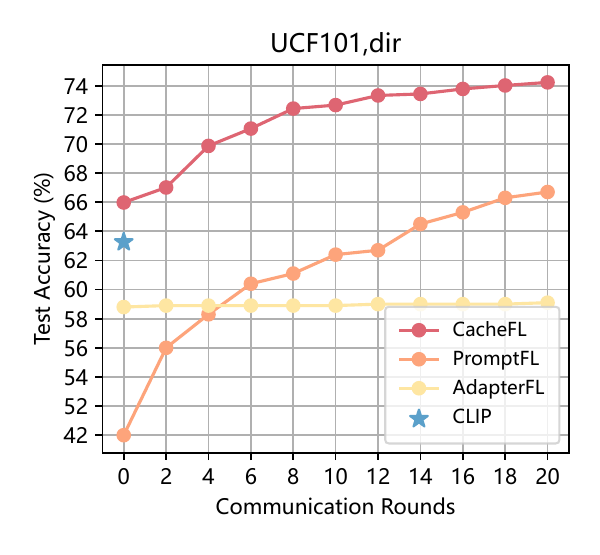}}
   \subfloat{\includegraphics[width=0.165\textwidth]{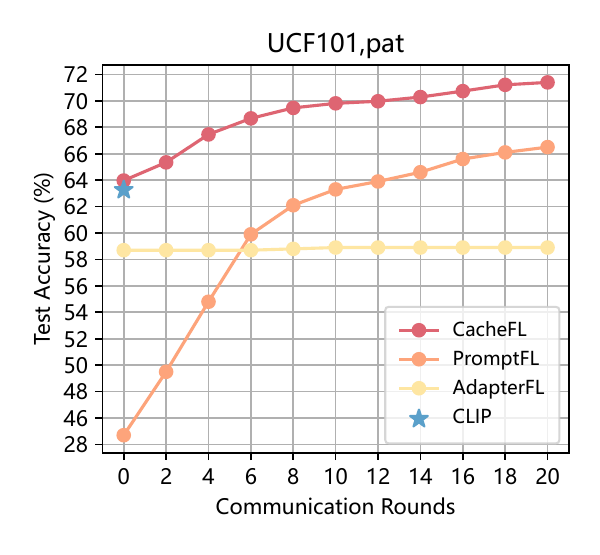}}
   \caption{Convergence Speed. The y-axis shows test accuracy (\%), and the x-axis shows communication rounds. Across 4 datasets (ImageNet, Stanford Cars, SUN397, and UCF101) and data settings (iid, dir, pat), CacheFL consistently outperforms all other methods in terms of final accuracy and convergence speed.}
    \label{fig:speed}
\end{figure*}

\begin{figure}[t]
   \centering
   \subfloat[\scriptsize Communication Overhead]{\label{fig:comm1}\includegraphics[width=0.22\textwidth]{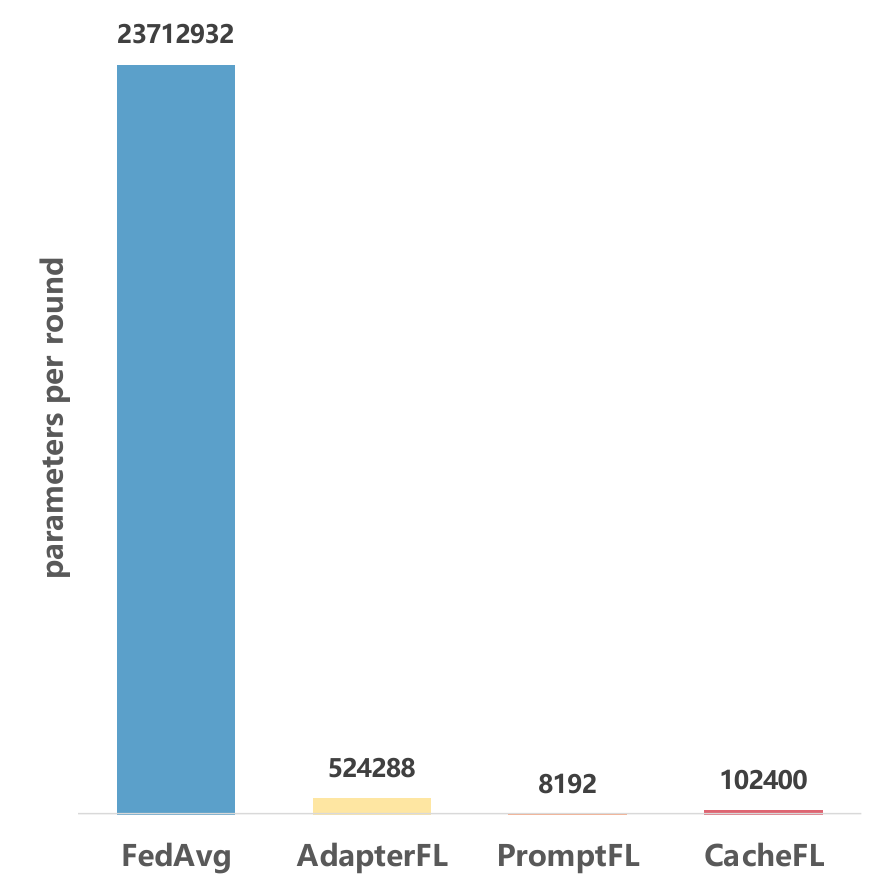}} \hfill
   \subfloat[\scriptsize Computation Overhead]{\label{fig:comm2}\includegraphics[width=0.22\textwidth]{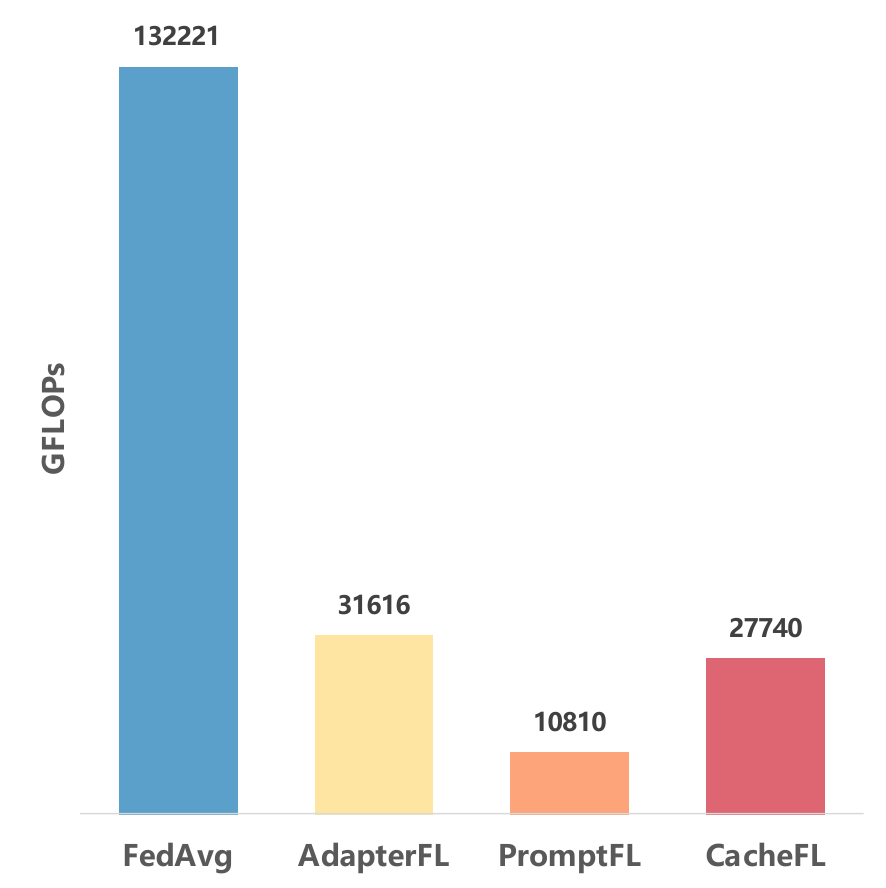}}
   \caption{Resource Consumption. We evaluate the communication cost based on the number of parameters uploaded per round and the computation cost based on the Floating Point Operations (FLOPs) required for training.}
    \label{fig:comm}
\end{figure}

\begin{figure}[t]
  \centering
  \includegraphics[width=\linewidth]{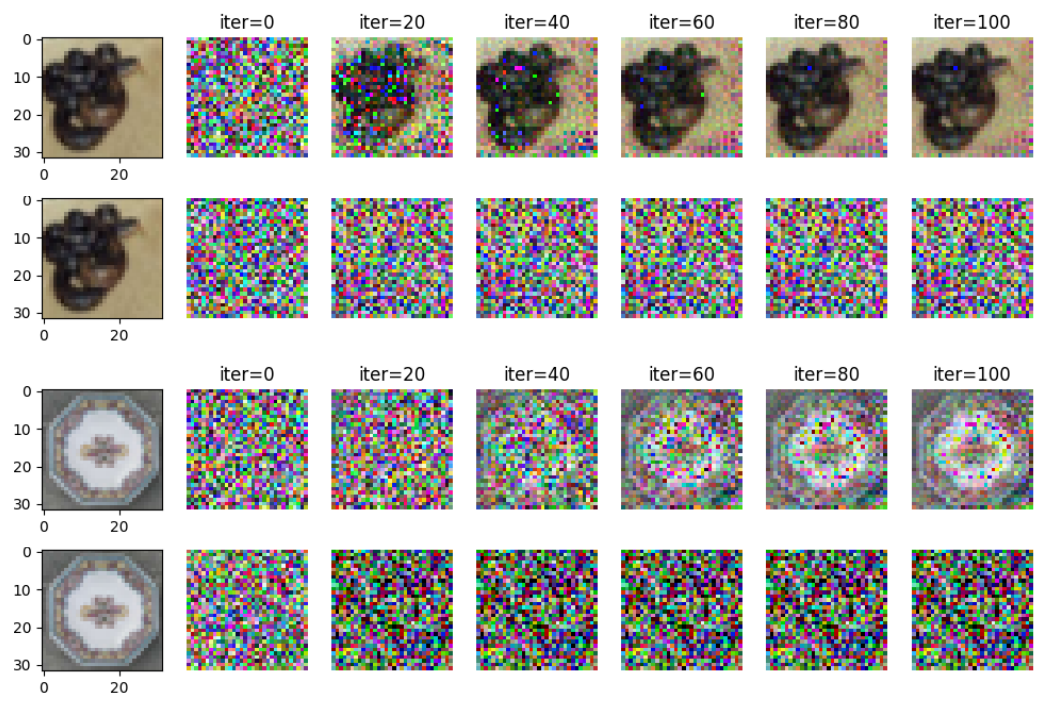}
  \caption{Gradient Inversion Attack. The top line of each image shows the results of FedAvg, where the attack almost recovers the raw image. The bottom line of each image presents the results of CacheFL, where the attack failed to reconstruct meaningful images.}
  \label{fig:dlg}
\end{figure}

\subsection{Overall Performance}
\noindent \textbf{Image Classification Performance.} 
Table \ref{tab:results} presents a comprehensive comparison of CacheFL with several representative baseline methods across multiple datasets and data distribution settings. As shown, CacheFL consistently achieves superior performance across all evaluated datasets, under both IID and non-IID scenarios. 
Moreover, CacheFL shows strong robustness to non-IID data distributions compared to PromptFL, which is evident in its performance on datasets such as DTD and Flowers102. For clients with imbalanced data, the class-balanced synthetic dataset effectively supplements missing knowledge during the initial stages, helping to mitigate the impact of heterogeneous data.
In comparison, the performance of traditional methods such as FedAvg and FedProx falls short across all datasets, primarily because these methods are not pre-trained and cannot achieve strong performance by training from scratch on a small amount of local data, especially considering the complexity of the datasets used. In contrast, methods like AdapterFL, PromptFL, and CacheFL improve upon FedAvg by incorporating pre-trained knowledge, leading to significantly better accuracy. 
However, AdapterFL also demonstrates a decline in performance compared to the zero-shot CLIP model. This indicates that simply integrating adapter tuning with federated learning does not guarantee improvements on new tasks. In scenarios with limited local data, conventional adapter fine-tuning methods may face difficulties in achieving convergence.
In contrast, CacheFL surpasses its competitors, especially under more extreme data imbalances. These results highlight CacheFL’s good adaptability and generalization capabilities across diverse datasets.

\noindent \textbf{Convergence Speed.}
As illustrated in Figure~\ref{fig:speed}, CacheFL consistently demonstrates superior performance across all data distribution settings. It achieves competitive accuracy in the early stages of training and exhibits rapid convergence, consistently surpassing other baseline methods in terms of test accuracy. 
Furthermore, the good convergence efficiency of CacheFL reduces the number of communication rounds needed in real-world training scenarios, thereby enhancing overall communication efficiency.

\noindent \textbf{Communication and Computation Overhead.}
Figure \ref{fig:comm1} illustrates the communication cost, quantified as the number of uploaded parameters per round, while Figure \ref{fig:comm2} depicts the computation cost, measured in Floating Point Operations (FLOPs) required for training. FedAvg imposes a significant communication burden of 23.7 million parameters per round and a training computation cost of 13,221 GB FLOPs, rendering it inefficient for large-scale federated learning applications. While PromptFL achieves the lowest communication and computation costs, it sacrifices some performance as seen in prior results. CacheFL, however, achieves an optimal balance between communication efficiency, computation efficiency, and model performance, demonstrating its scalability and practicality for federated learning in resource-constrained settings.

\noindent \textbf{Privacy Preservation.} 
CacheFL only requires class names as DALL·E prompt inputs, not actual data or its distribution, minimizing privacy risks. The generated outputs are used solely to initialize the cache model and do not influence the training process, ensuring that the privacy of local data is not leaked. 
Moreover, we evaluate the privacy-preserving capabilities of CacheFL by applying a representative gradient inversion attack to reconstruct raw training images, following the method outlined by Zhu \emph{et al.} \cite{zhu2019deep}. To demonstrate the effect of defense, we compare the CacheFL approach with the FedAvg method. We utilize dummy data and labels and train for 100 iterations, recording reconstruction results at intervals of 20 epochs, as shown in Figure \ref{fig:dlg}. The findings demonstrate that, with FedAvg, raw images can be almost recovered after a certain number of epochs, whereas the reconstructed images from CacheFL resemble noise. This indicates that CacheFL effectively preserves privacy through its cache model mechanism.

\subsection{Ablation Study}

\begin{figure*}[t]
   \centering
   \subfloat[ImageNet]{\includegraphics[width=0.22\textwidth]{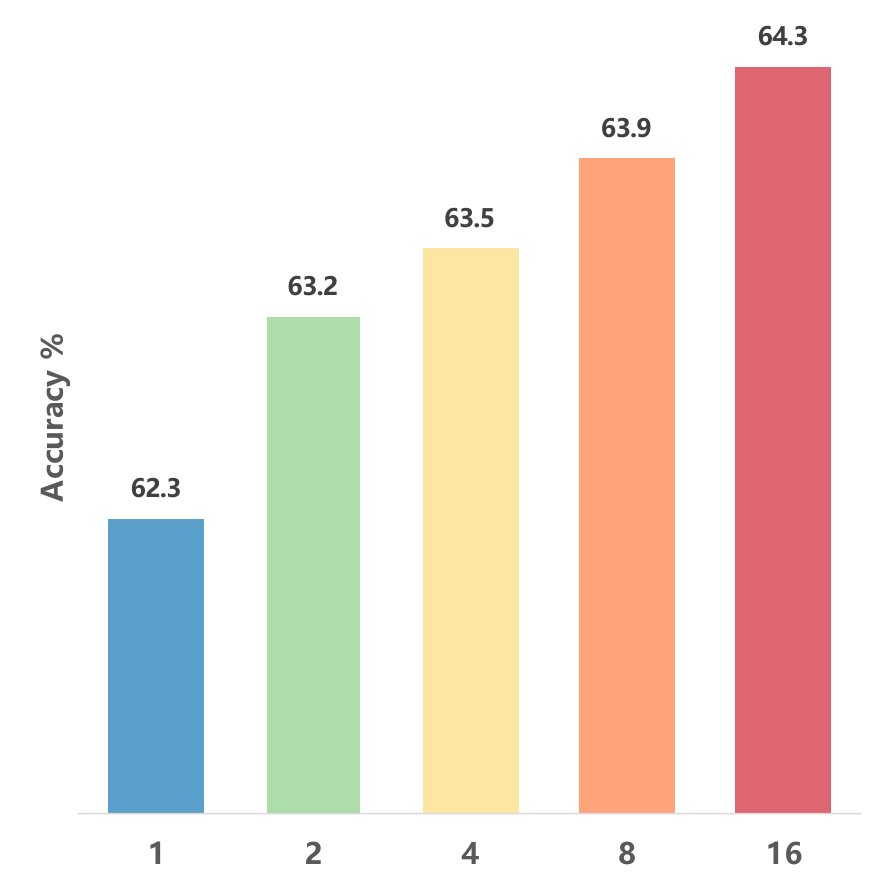}}
   \subfloat[StandfordCars]{\includegraphics[width=0.22\textwidth]{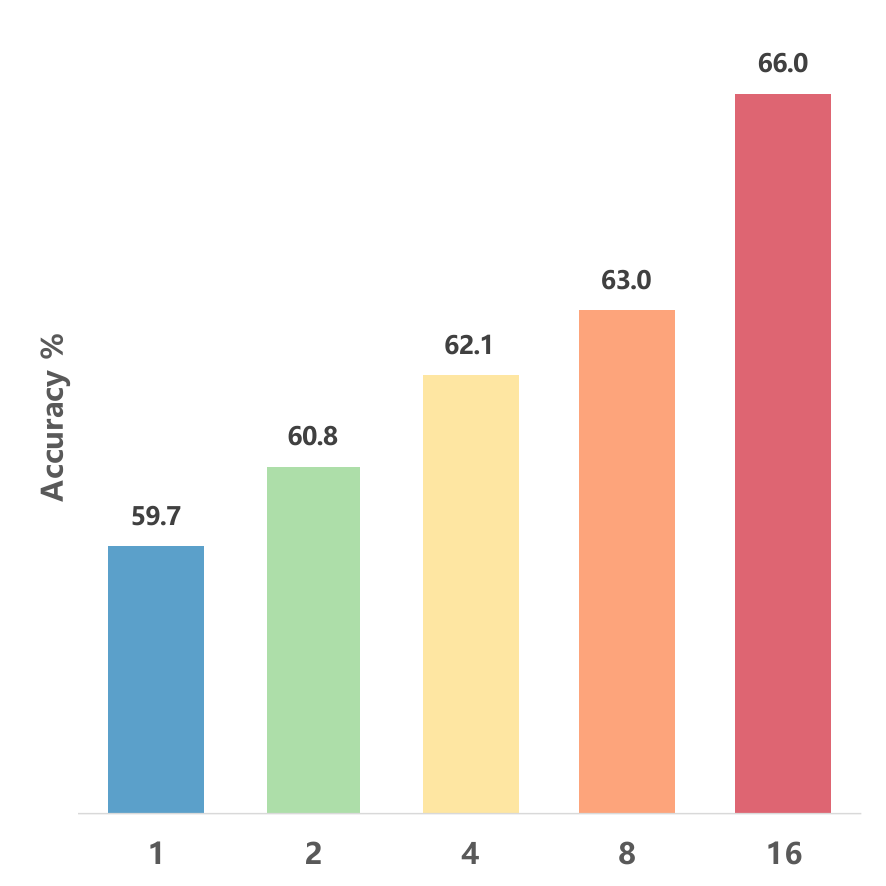}}
   \subfloat[SUN397]{\includegraphics[width=0.22\textwidth]{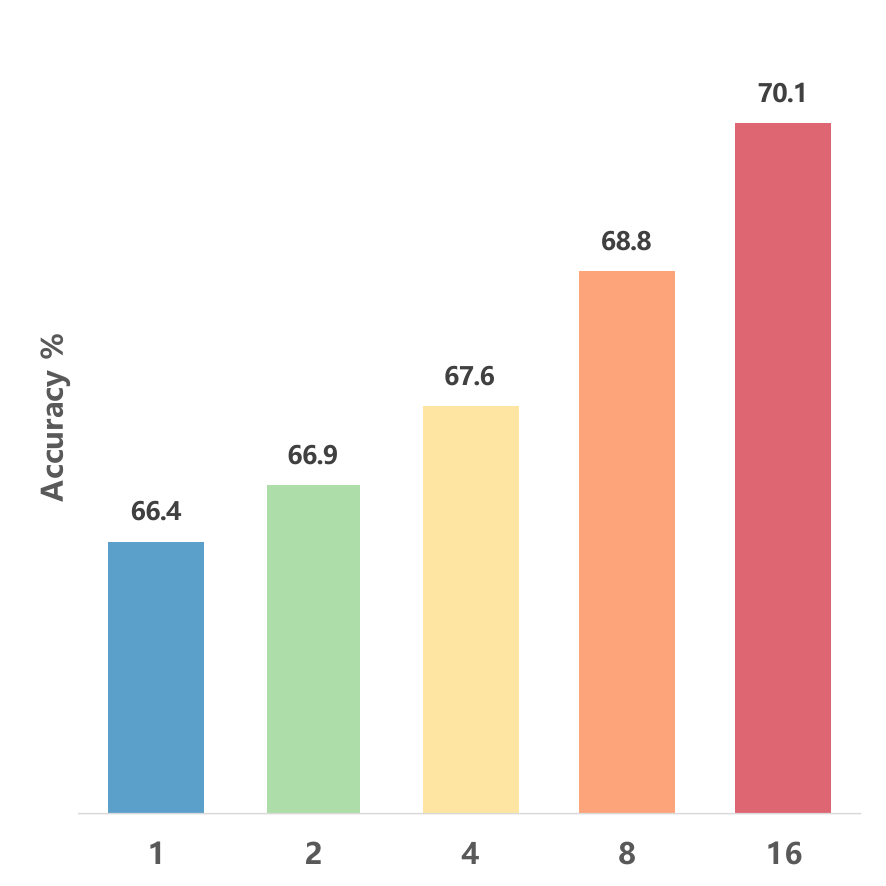}}
   \subfloat[UCF101]{\includegraphics[width=0.22\textwidth]{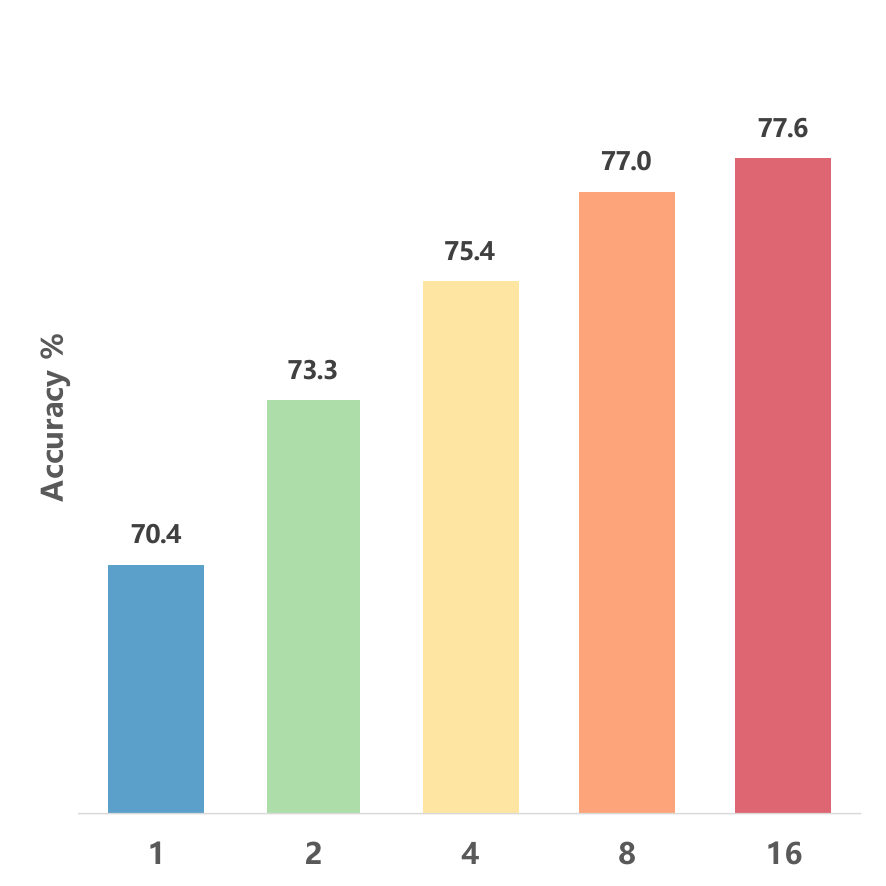}}
   \caption{Performance with the Number of Synthesized Images. DALL·E generates images for each category in increments of 1 to 16.}
    \label{fig:num}
\end{figure*}

\begin{table}[t]
	\caption{Performance with Synthetic Dataset Initialization and Federated Training.}
    \centering
    \resizebox{0.49\textwidth}{!}{
    \begin{tabular}{lllll}
    \toprule
	& ImageNet & StandfordCars & SUN397 & UCF101 \\
    \midrule
    PromptFL & 62.1     & 58.7  & 64.2   & 67.3   \\
	\textbf{CacheFL} & \textbf{64.3}     & \textbf{66.0}          & \textbf{70.1}   & \textbf{77.6} \\ 
    w.o. Syn & 62.6  & 59.7  & 65.6   & 69.0   \\
    w.o. Training & 60.3   & 55.2  & 60.4   & 57.9   \\
	\bottomrule
    \end{tabular}
    }
    \vspace{2pt}
    \begin{flushleft}``w.o. Syn" represents random cache model initialization without a synthetic dataset. ``w.o. Training" represents the method of not using federated training and only using the initial cache model with a synthetic dataset.\end{flushleft}
	\label{tab:fedtr}
\end{table}

\noindent \textbf{Performance with Synthetic Dataset Initialization.}  
Table \ref{tab:fedtr} compares the performance of cache models initialized randomly versus those initialized with a synthetic dataset across four different datasets. Performance degrades with random initialization, indicating that the class-balanced synthetic dataset effectively supplements missing knowledge during the initial stages. Existing baselines struggle to leverage synthetic data without substantial modifications. In contrast, our approach requires only a lightweight feature extraction and encoding step to initialize the cache model using a small synthetic dataset. Remarkably, our method achieves competitive results relative to PromptFL even when using random initialization. These findings indicate that although our approach benefits from synthetic data, it does not depend exclusively on it.
Moreover, the results indicate that CacheFL maintains strong performance even under worst-case conditions, where biases and domain gaps between synthetic and real data make the initialization comparable to random. As these discrepancies are reduced, the initialization becomes more effective, resulting in improved outcomes.

\noindent \textbf{Performance with Federated Training.} 
Table \ref{tab:fedtr} also compares the performance of CacheFL with and without federated training across four datasets. CacheFL consistently outperforms the version without federated training on all datasets. Notably, the federated version achieves an accuracy of 77.6\% on UCF 101, marking a substantial improvement of 19.7\%. This significant improvement demonstrates the effectiveness of federated training in enhancing model generalization and performance by leveraging private data. The results suggest that federated training plays a critical role in maximizing the potential of CLIP.

\noindent \textbf{Performance with the Number of Synthesized Images.} 
Figure \ref{fig:num} presents the impact of the number of synthesized images generated by the DALL·E model on the accuracy of the CacheFL method on four datasets. As the number of synthesized images increases from 1 to 16 for each category, there is a clear upward trend in accuracy. The experimental results demonstrate that the more synthetic data samples we cache, the higher the accuracy CLIP with the cache model can achieve. However, the rate of improvement diminishes slightly beyond 8 images on the UCF101 datasets, indicating potential diminishing returns. This result suggests that while increasing the number of synthesized images can enhance performance, doing so indefinitely may not lead to continuous gains and could instead result in higher computational costs. Therefore, identifying an optimal range for balancing performance improvements with computational efficiency is important. 

\noindent \textbf{Performance with Hyperparameters $\alpha$ and $\beta$.} 
Table \ref{hyper} presents the model’s performance across varying values of the hyperparameters $\alpha$ and $\beta$. To isolate the effect of each hyperparameter, we conduct controlled experiments: $\alpha$ is varied from 0.0 to 2.0 while keeping $\beta$ fixed at 1.0, and $\beta$ is varied over the same range with $\alpha$ fixed at 0.5. The results indicate that the model achieves higher accuracy when $\alpha = 0.5$ and $\beta = 1.0$. For $\alpha$, performance improves initially but drops sharply beyond 1.0, suggesting that while incorporating the logit from the cache model is beneficial, excessive influence may hinder the performance of CLIP in this setting. In contrast, the model is more robust to variations in $\beta$, with accuracy remaining relatively stable. This indicates that while both hyperparameters influence performance, tuning $\alpha$ is particularly critical for achieving optimal results.

\begin{table}[t]
\centering
\caption{Performance under different values of hyperparameters $\alpha$ and $\beta$.}
\begin{tabular}{c|ccccc}
\toprule
\textbf{$\alpha$} & 0.0 & 0.5 & 1.0 & 1.5 & 2.0 \\
\midrule
accuracy & 61.5 & 64.1 & 63.1 & 56.8 & 51.4 \\
\midrule
\textbf{$\beta$} & 0.0 & 0.5 & 1.0 & 1.5 & 2.0 \\
\midrule
accuracy & 62.0 & 63.5 & 64.1 & 64.0 & 62.8  \\
\bottomrule
\end{tabular}
\begin{flushleft}We vary $\alpha$ from 0.0 to 2.0 while fixing $\beta$ at 1.0, and vary $\beta$ from 0.0 to 2.0 while fixing $\alpha$ at 0.5.\end{flushleft}
\label{hyper}
\end{table}

\section{Conclusion}\label{5}
In this work, we propose CacheFL, a novel framework that integrates adapter tuning with federated learning to fine-tune CLIP and similar vision-language models. CacheFL preserves data privacy and security while utilizing client-side knowledge to enhance classification performance. By transmitting only the lightweight cache model rather than the full CLIP model, CacheFL significantly reduces both communication and computation costs during the federated learning process. Furthermore, it mitigates the challenges of data heterogeneity by employing class-balanced datasets generated by DALL·E and leveraging CLIP’s strong generalization ability. Experimental results highlight CacheFL’s promise as a practical, privacy-preserving, and efficient solution for federated fine-tuning of vision-language models.
In future work, we aim to improve the quality of the generated datasets, as their construction directly influences the cache model's performance. For example, incorporating GPT \cite{achiam2023gpt} models to generate more diverse text descriptions could result in higher-quality images, further boosting CacheFL’s performance and applicability. Additionally, we plan to explore adaptive hyperparameter adjustment techniques as an alternative to empirical settings and grid search.

\section*{Acknowledgments}
This work was supported in part by the STI 2030-Major Projects of China under Grant 2021ZD0201300, and by the National Science Foundation of China under Grant 62276127.

\bibliographystyle{IEEEtran}
\bibliography{main.bib}

\end{document}